\documentclass[1p,11pt]{elsarticle}
\makeatletter
\def\ps@pprintTitle{%
 \let\@oddhead\@empty
 \let\@evenhead\@empty
 \def\@oddfoot{\centerline{\thepage}}%
 \let\@evenfoot\@oddfoot}
\makeatother

\usepackage{tikz,amsthm,verbatim,relsize,amssymb,amsmath,epstopdf,upgreek,bbm,graphicx,fancybox,color,soul,setspace,lineno,colortbl,longtable,lineno,hyperref,amsmath,amsfonts,amssymb,color,proof,colortbl,mathbbol,multirow,dsfont,soul,fixltx2e,graphicx,esvect,colonequals,rotating,tabularx,booktabs,tikz,siunitx,soul,verbatim,dsfont,pgfplots,cleveref,lineno,hyperref,hyphenat,url,booktabs,lscape,mathrsfs,amsthm,pifont,subfloat,tikz,algorithm,algpseudocode,algorithmicx,subfigure}

\usetikzlibrary{arrows,calc}

\tikzstyle{state}=[circle,thick,draw=black!75,fill=black!20,minimum size=6mm, text=black]

\usepackage{pgfplotstable}
\pgfplotsset{compat=1.7,
        /pgfplots/ybar legend/.style={
        /pgfplots/legend image code/.code={%
        \draw[##1,/tikz/.cd,bar width=3pt,yshift=-0.2em,bar shift=0pt]
                plot coordinates {(0cm,0.8em)};},
},
}

\newcommand{\mc}{\mathcal}
\renewcommand{\vec}[1]{\boldsymbol{#1}}

\newcommand{\dfn}{\colonequals}
\newcommand{\clhead}{\tt{C}}
\newcommand{\destination}{\tt{Rqs}}
\newcommand{\attacker}{\tt{A}}
\newcommand{\noos}{\Lambda}
\newcommand{\os}{\lambda}

\newcommand{\data}{\tt{D}}

\newcommand{\maximinA}{\mixedstratA^{\dagger}}
\newcommand{\maximinD}{\mixedstratD^{\dagger}}

\newcommand{\1}{\langle}
\newcommand{\2}{\rangle}

\newcommand{\reply}{\textsc{Reply}}

\newcommand{\SSE}{\mathrm{SSE}}
\newcommand{\NE}{\mathrm{NE}}
\newcommand{\BR}{\mathrm{BR}}

\newcommand{\mixed}{(\mixedstratD,\mixedstratA)}
\newcommand{\mixedNE}{\mixedstratD^{\NE},\mixedstratA^{\NE}}
\newcommand{\setsse}{\Omega^{\SSE}}
\newcommand{\setne}{\Omega^{\NE}}
\newcommand{\setmaximin}{\Omega^{\mathrm{maximin}}}

\newcommand{\nzs}{\Gamma}
\newcommand{\zs}{\Gamma_0}

\newcommand{\mixedstratD}{\vec{\rho}}
\newcommand{\mixedstratA}{\vec{\mu}}
\renewcommand{\arraystretch}{1.5}

\definecolor{Gray}{gray}{0.88}
\newtheorem{thm}{Theorem}
\newtheorem{lemma}{Lemma}

\newtheorem{definition}{Definition}

\definecolor{halfgray}{gray}{0.75}

\newcolumntype{P}[1]{>{\centering\arraybackslash}p{#1}}
\newcolumntype{M}[1]{>{\centering\arraybackslash}m{#1}}

\usetikzlibrary{arrows,automata}
\usetikzlibrary{positioning}
\usetikzlibrary{shapes}

\tikzset{
    >=stealth',
    punkt/.style={
           ellipse,
           rounded corners,
           draw=black, 
           text width=2.5em,
           minimum height=2em,
           text centered},
    pil/.style={
           ->,
           shorten <=2pt,
           shorten >=2pt,},
    pil1/.style={ellipse,draw = white}
}

\journal{Ad Hoc Networks}
\begin{document}
\begin{frontmatter}

\title{Game Theoretic Path Selection to Support Security in Device-to-Device Communications}

\author[uob]{Emmanouil Panaousis}
\author[cpt]{Eirini Karapistoli}
\author[uog]{Hadeer Elsemary}
\author[uom]{Tansu Alpcan}
\author[qm]{MHR Khuzani}
\author[uomg]{Anastasios A.~Economides}

\address[uob]{University of Brighton, UK} 
\address[cpt]{Capritech Limited, UK}
\address[uog]{University of Gottingen, Germany}
\address[uom]{University of Melbourne, Australia}
\address[qm]{Queen Mary University of London, UK}
\address[uomg]{University of Macedonia, Greece}

\begin{abstract}
\footnote{\copyright$\langle$2016$\rangle$. This manuscript version is made available under the CC-BY-NC-ND 4.0 license http://creativecommons.org/licenses/by-nc-nd/4.0/. \\
DOI: 	10.1016/j.adhoc.2016.11.008.}
Device-to-Device (D2D) communication is expected to be a key feature supported by 5G networks, especially due to the proliferation of Mobile Edge Computing (MEC), which has a prominent role in reducing network stress by shifting computational tasks from the Internet to the mobile edge.~Apart from being part of MEC, D2D can extend cellular coverage allowing users to communicate directly when telecommunication infrastructure is highly congested or absent.~This significant departure from the typical cellular paradigm imposes the need for decentralised network routing protocols.~Moreover, enhanced capabilities of mobile devices and D2D networking will likely result in proliferation of new malware types and epidemics.~Although the literature is rich in terms of D2D routing protocols that enhance quality-of-service and energy consumption, they provide only basic security support, e.g., in the form of encryption.~Routing decisions can, however, contribute to collaborative detection of mobile malware by leveraging different kinds of anti-malware software installed on mobile devices.~Benefiting from the cooperative nature of D2D communications, devices can rely on each other's contributions to detect malware.~The impact of our work is geared towards having more malware-free D2D networks.~To achieve this, we designed and implemented a novel routing protocol for D2D communications that optimises routing decisions for explicitly improving malware detection.~The protocol identifies optimal network paths, in terms of malware mitigation and energy spent for malware detection, based on a \emph{game theoretic model}.~Diverse capabilities of network devices running different types of anti-malware software and their potential for inspecting messages relayed towards an intended destination device are leveraged using game theoretic tools.~An optimality analysis of both Nash and Stackelberg security games is undertaken, including both zero and non-zero sum variants, and the Defender's equilibrium strategies.~By undertaking network simulations, theoretical results obtained are illustrated through randomly generated network scenarios showing how our protocol outperforms conventional routing protocols, in terms of expected payoff, which consists of: \emph{security damage inflicted by malware} and \emph{malware detection cost}.
\end{abstract}


\end{frontmatter}

\section{Introduction}
Demand for anytime-anywhere wireless broadband connectivity and increasingly stringent Quality of Service (QoS) requirements pose new research challenges.~As mobile devices are capable of communicating in both cellular (e.g.~4G) and unlicensed (e.g.~IEEE 802.11) spectrum, the Device-to-Device (D2D) networking paradigm has the potential to bring several immediate gains.~Networking based on D2D communication \cite{Feng:IEEEComms:2014,Nishiyama:IEEEComms:2014,Tehrani:IEEEComms:2014,Fodor:IEEEComms:2012,Doppler:IEEEComms:2009} not only facilitates wireless and mobile peer-to-peer services, but also provides energy efficient communications, locally offloading computation, offloading connectivity, and high throughput.~The most emerging feature of D2D is the establishment and use of multi-hop paths to enable communications among non-neighbouring devices.~In multi-hop D2D communications, data are delivered from a source to a destination via intermediate (i.e.~relaying) devices, independently of operators' networks.~

\subsection{Motivation}
To motivate the D2D communication paradigm, we emphasise the need for \emph{localised applications}.~These run in a collaborative manner by groups of devices at a location where telecommunications infrastructures: (i) are not present at all, e.g.~underground stations, airplanes, cruise ships, parts of a motorway, and mountains; (ii) have collapsed due to physical damage to the base stations or insufficient available power, e.g.~areas affected by a disaster such as earthquake; or (iii) are over congested due to an extremely crowded network, e.g.~for events in stadiums, and public celebrations.~Furthermore, relay by device can be leveraged for commercial purposes such as advertisements and voucher distributions for instance in large shopping centres.~This is considered a more efficient way of promoting businesses than other traditional methods such as email broadcasting and SMS messaging due to the immediate identification of the clients in a surrounding area.~Home automation and building security are another two areas that multi-hop data delivery using D2D communications is likely to overtake our daily life in the near future while multi-hop D2D could be also leveraged towards the provision of anonymity against cellular operators \cite{Ardagna:IEEETSC:2014}.

A key question related to multi-hop D2D networks is, \emph{which route should the originator of some data choose to send it to an intended destination?}.~This has been exhaustively investigated in the literature of wireless and mobile ad hoc routing with well-known protocol to be among others AODV \cite{Perkins:2003:AODV:RFC3561}, DSR \cite{Johnson:2001:DSR}, and OLSR \cite{Clausen:2003:OLSR:RFC3626}.~A thorough survey of standardisation efforts in this field has been published by Ramrekha et al.~\cite{RPP13}.

Due to the myriad number of areas D2D communications are applicable to, devices are likely to be an ideal target for attackers who aim to infect devices with malware.~Authors in \cite{malwaresurvey2014} point out that malware in current smartphones and tablets have recently rocketed and established its presence through advanced techniques that bypass security mechanisms of devices.~Malware can spread, for instance, through a Multimedia Messaging System (MMS) with infected attachments, or an infected message received via Bluetooth aiming at stealing users' personal data or credit stored in the device.~An example of a well-known worm that propagates through Bluetooth was Cabir, which consists of a message containing an application file called {\tt caribe.sis}.~Apart from malware infection, Khuzani et al.~\cite{armanmalwareepidemics} have investigated outbreaks of malware (i.e.~malware epidemics) mainly by adopting the notion of D2D communication.~Finally, social engineering attacks against mobile phones is one of the most serious threats, as presented in a relevant survey here \cite{heartfield2016taxonomy}.~For thorough surveys on mobile malware one may refer to \cite{malwaresurvey2014,LaPolla:IEEECST:2013}.

\subsection{Innovation}
In a nutshell,this paper presents a novel routing protocol, for D2D communications, that supports malware detection in an optimal way by using non-cooperative \emph{game theoretic} tools, which have been extensively used in the security literature (e.g.~\cite{Alpcan:Book:2012}) and in D2D routing (e.g.~\cite{Naserian:AdHoc:2009}).~Game theory has also been used for other than routing purposes~\cite{xiao2015bayesian}, \cite{Long:AdHoc:2011,Wang:AdHoc:2009} in D2D networks.~In this paper we only focus on security games and we tackle a decision-making routing challenge, in D2D networks, in presence of an adversary who injects malware into the network, after she has compromised a gateway that connects the D2D network with the cloud.~This assumption is fairly realistic given the vast power attackers have in their hands these days to successfully exploit vulnerabilities of modern gateways.~Our underlying network has been inspired by the \emph{Mobile Edge Computing} (MEC) (also refer to as Fog Computing) paradigm as a step towards addressing security within the realm of an increasingly important area of 5G.

Our protocol, called $i$Routing (abbreviating ``intelligent Routing''), is designed upon the theoretical analysis of a simple yet illuminating two-player security game between the \emph{Defender}, which abstracts a D2D network, and the \emph{Attacker}, which abstracts any adversarial entity that wishes to inject malware into the D2D network.~We have proven that the Defender's \emph{equilibrium strategies} leave the network better off, in terms of \emph{expected payoff}, which is a combination of \emph{security damage} and \emph{malware detection cost} (i.e.~cycles process units).~Note that $i$Routing can work on top of underlying physical and MAC layer protocols \cite{Jianting:IEEEComms:2013,Daohua2014}.

It is worth noting that this paper does not tackle secure routing issues in traditional ways.~For a survey of secure routing protocols for wireless ad hoc networks, see \cite{Abusalah:IEECST:2008,Gupte:AdHoc:2003}.~Such protocols mainly aim at enabling confidentiality, and integrity of the communicated data and they do not consider underlying collaborative malware detection.

\subsection{Progress beyond relevant work}
This paper extends, in a significant manner, the results initially presented in \cite{Panaousis:Gamesec:2014}.~The exact differences are summarized below.

\begin{itemize}
\item \cite{Panaousis:Gamesec:2014} assumes a pure device-to-device network while in this paper the device-to-device network has been enriched with a part of mobile edge computing.~The network devices request services from the MEC server and multi-hopping enables communication between the MEC server and the different devices to overcome proximity issues due to the latter being outside the transmission range of the server.~In this paper, the security challenge is how to safely utilise MEC services where a cluster-head (i.e.~MEC server) might be compromised by an adversary.~Although this does not introduce any new challenge in terms of malware detection and routing, it is an assumption that places the idea of the paper within mobile edge computing and 5G architectures.
	
\item This paper assumes different mobile operating systems and these can be infected with different types of malware as opposed to \mbox{\cite{Panaousis:Gamesec:2014}}, which goes as far as considering just a set of malicious messages that are sent from the attacker's device to infect the legitimate devices.~This also has the effect of defining, in this paper, the Malware Detection Game whereas in \cite{Panaousis:Gamesec:2014}, the defined game is called Secure Message Delivery Game.
	
\item In \cite{Panaousis:Gamesec:2014}, a confusion matrix is defined to determine how the different devices of the network can detect malicious messages.~In this paper here we take a more realistic, in the terms of cyber security, approach where for each device there is a probability to be compromised by malware.~Therefore, each route has, in turn, a penetration level, which is the probability the route to be compromised due to one or more devices on it being vulnerable.
	
\item In \cite{Panaousis:Gamesec:2014}, the details about the interdependencies of malicious message detectors is not discussed, while in our paper here we explicitly say that each control detects different signs of malware and \emph{no interdependencies}, in terms of detection capabilities, are assumed, i.e.~we have assumed that an anti-malware control is the minimal piece of software that detects certain malicious signs.
	
\item In \cite{Panaousis:Gamesec:2014}, the Attacker is not assumed to monitor the network before launching a malware attack (no reconnaissance) while in our paper here the Attacker surveils the network before injecting malware giving us a Stackelberg game to study.
	
\item In \cite{Panaousis:Gamesec:2014}, only Nash Equilibria (NE) and maximin strategies have been studied.~On the other hand, our paper here derives Strong Stackelberg Equilibria (SSE) and shows the relationship among three of them; SSE, NE and maximin.~Not only that, but this paper exhibits much larger depth of mathematical analysis referring also to best responses of players.~Finally, it proves the equality of strategies of different games, such zero-sum and non-zero sum across all strategic types (Nash, Stackelberg, maximin).
	
\item Although Panaousis et al. \cite{Panaousis:Gamesec:2014} has investigated both zero sum and non-zero sum games, where in the latter the utility of the Attacker is a positive affine transformation (PAT) of the defender's utility, in this paper we go beyond that.~We show the equality of the different strategies holds in a more generic (i.e.~than the PAT case) payoff structure where the Attacker’s utility is a strictly positive scaling of the Defender's utility.	
	
\item All simulations in \cite{Panaousis:Gamesec:2014} were numeric; as well as they do not compare the performance of the proposed routing protocol with other device-to-device routing protocols.~For the purposes of our paper here we have undertaking a network simulation to compare the proposed protocol with legacy routing protocols using the OMNeT++ network simulator.~In this way we have simulated physical and link-layer network characteristics.

\item In our paper here, we have considered, in our simulations, the efficacies of some of the most-recent real-world anti-malware controls against real-world malware types as opposed to the purely numeric assignment to the different variables.
	
\item In our simulations here, we have included a new Attacker type, called Weighted, which allows the adversary to distribute her resources proportionally, over the different routes, aiming at the highest expected damage.~This type of Attacker was not simulated in \cite{Panaousis:Gamesec:2014}.
\end{itemize}

\subsection{Main assumptions}
Our analysis assumes that each device has some malware detection capabilities (e.g.~anti-malware software).~Therefore, a device is able to detect malicious application-level events.~In other words, each device has its own detection rate which contributes towards the overall detection rate of the routes that this device is part of.~In order to increase malware detection, the route with the highest detection capabilities must be selected to relay the message to the destination.

However, due to the different malware types available to attackers, these days, such a decision is not trivial.~One could argue that if we know the probability of a malware type to be chosen, we can develop a proportional routing strategy, which will distribute security risks across the different routes by choosing routes in a proportional, to their malware detection capabilities, manner.~Since this knowledge can not be taken for granted in addition to the volatile nature of such statistics, in this paper we use game theory to optimise routing decisions to support malware detection in D2D networks, regardless of the probability of the different malware to be used by the Attacker.~

\subsection{Outline}
The remainder of this paper is organised as follows: In Section \ref{related_work}, we review related work with more emphasis to be given in papers at the intersection of game theory, security, and routing for wireless ad hoc networks (i.e.~prominent example of D2D networking).~In Section \ref{system}, we present the system and game models, while in Section \ref{solutions}, we devise game solutions.~In Section \ref{analysis}, we undertake optimality analysis which leads to a list of theoretic contributions.~Section \ref{irouting} describes, in detail, the $i$Routing protocol, and in Section \ref{simulations}, we compare $i$Routing against other routing protocols.~Finally, Section \ref{conclusion} provides concluding remarks and points towards future research.

\section{Related work}
\label{related_work}

In this section,~we briefly review the state-of-the-art,~in chronological order,~in terms of game theoretic approaches at the intersection of three fields: security,~routing,~and device-to-device networks.~Another set of game theoretic works that focus on optimising intrusion detection strategies per se than adjusting routing decisions to optimally support intrusion detection,~consist of papers such as \cite{Patcha:IWIA:2004},~\cite{Liu:GAMENETS:2006},~\cite{Liu:IJSN:2006},~\cite{Liu:IJSN:2006},~\cite{Marchang:ADCOM:2007},~\cite{Otrok:ICDCSW:2009},~\cite{Santosh:ADCOM:2008},~and \cite{Cho:IEEEJR:2010}.~Our work is complementary to this literature as it optimises end-to-end path selections,~in terms of malware detection efficacy and computational effort.

Looking more into decision regarding packet forwarding by using game theoretic tools and without incentive mechanisms in place,~Felegyhazi et al.~\cite{Felegyhazi:TMC:2006}~have studied the Nash equilibria of packet forwarding strategies with tit-for-tat punishment strategy in an iterative game.~In each stage (i.e.~time slot) of the game,~each device selects its cooperation level based on the normalised throughput it experienced in the previous stage.~As opposed to $i$Routing,~the authors do not propose a new end-to-end routing protocol; instead they consider a shortest path algorithm.~Also,~they assume the existence of internal malicious or selfish nodes in contrast to our work here,~which models an adversary outside of the D2D cluster,~who aims to infect legitimate devices with malware.

In a more security-oriented vein,~Yu et al.~\cite{Yu:IEEEJMC:2007} have used game theory to study the dynamic interactions,~in mobile ad hoc (device-to-device) networks,~between ``good'' nodes,~which initially believe that all other nodes are not malicious,~and ``adversaries'',~which are aware of which nodes are good.~They propose secure routing and packet forwarding games that consist of 3 stages: route participation; route selection; and packet forwarding.~In the first stage,~a node decides whether to be part of route or not; in the second phase,~a node who wishes to send a packet to  a destination,~after it discovers a \emph{valid route} (called when all nodes agree to be part of it),~it either uses the discovered route or not; and,~finally,~in the third phase,~each relay node decides to forward or not an incoming packet.~They have derived optimal defence strategies and studied the maximum potential damage,~which incurs when attackers find a route with maximum number of hops and they inject malicious traffic into it.~The same authors also combined this game with a secure routing game but without considering noise and imperfect monitoring.~Yu et al.~\cite{Yu:IEEEJIFS:2007} extended \cite{Yu:IEEEJMC:2007} and proposed a secure cooperation game under noise and imperfect monitoring.~Likewise,~Yu and Liu tackled the same challenge and presented a richer set of performance evaluation results in \cite{Yu:IEEEJIFS:2008}.~The above publications do not tackle the same challenge with $i$Routing,~as they do not investigate the selection of a route among an available set of routes to deliver packets from a source to a destination

Finally,~in \cite{Panaousis:LCN:2009},~Panaousis and Politis present a routing protocol that respects the energy spent by intrusion detection on each route and therefore prolonging network lifetime.~This paper takes a simple approach,~according to which the attacker either attacks or not a route,~and the Defender,~likewise,~decides whether to allocate resources to defend or not.

None of the aforesaid protocols consider the propagation of malware within the network and none of these works investigates Stackelberg games, which basically assume that the Attacker conducts surveillance before deciding upon her strategy.~This is a reasonably realistic assumption when looking at the intelligence of cyber hackers and it is a conventional decision in other security related fields \cite{korzhyk2011,tambe2011,wang2013,Kar:AAMAS:2015}.

\section{System description and game model}
\label{system}

This section presents our underlying system model along with its components.~Mobile-edge computing (MEC) is an emerging paradigm that allows mobile applications to offload computationally intensive workloads to a MEC server.~This introduces a new network architecture concept that provides  cloud-computing capabilities at the edge of the mobile network.~The MEC server is likely to be setup by a service provider to ensure that it can provide a service environment with very low latency and high-bandwidth.

\subsection{System description}

We use a motivational paradigm demonstrating how D2D communication can be combined with a MEC architecture \cite{bonomi2012fog}, as depicted in Fig.\,\ref{fig:network}.~In our model, MEC is an intermediate layer between a \emph{D2D cluster} and the \emph{cloud}, aiming at \emph{low-latency service delivery} from the latter to the former, and it can serve users by using local short-distance high-rate connections.~The intermediate layer can contain a number of deployed MEC servers aiming to handle the localised requests issued by cluster users.

We assume that devices within a cluster can communicate in a D2D manner: directly or by using multi-hop routes.~The cluster is formed based on discovery protocols that run in each device. These allow to sense the environment and create a list of one-hop neighbours in order to be able to communicate should any request to forward data or a direct request be sent. We also assume no cellular infrastructure within the cluster, which means that devices can only communicate in a device-to-device fashion.

It is envisaged that such scenarios will be very common in 5G ecosystems where heterogeneous wireless technologies (e.g.~NB-LTE, WiFi, ZigBee, Bluetooth) will facilitate D2D communication \cite{Tehrani:IEEEComms:2014}.~For example, a device that seeks some data, can request this from other devices in its cluster, and if the \textsc{Request} cannot be served the MEC servers must be contacted to assist with the discovery of this data.

The idea here is that a MEC server is dedicated to provide predefined service applications to cluster users without the need to communicate with the cloud so that it accelerates responses while \emph{``pushing'' the cloud away of the user}.~We assume that each D2D cluster has a \emph{cluster-head} \cite{Asadi:IEEE:2014}, which is a device that communicates with the MEC servers. The main functionalities of a cluster-head are (i) to forward the \textsc{Request} of a device to the MEC servers, and (ii) upon its response, to transmit the \textsc{Reply} back to the requestor.~In this work, the cluster-head can be any device of the cluster. The MEC server is expected to talk to both the cloud servers and the cluster-head to handle functionalities such as device identifier allocation, call establishment, UE capability tracking, service support, and mobility tracking. Note that the election of the cluster-head is not investigated in this paper and also this paper is not concerned about deciding the nature of the cluster-head.

\subsection{Adversarial model}
As any open wireless environment, akin to one described in this paper, can be a target of adversaries.~More specifically, in this paper, we assume the existence of a malicious device, called \emph{the Attacker}, that can launch a Man-In-the-Middle (MITM) attack by \emph{hijacking the link} between the cluster-head and MEC servers.~Our analysis adopts the Dolev-Yao model \cite{dolev83}.~According to this, the D2D network, along with its established connection with the MEC servers, is represented as a \emph{set of abstract entities} that exchange messages.~Yet, the adversary is capable of overhearing, intercepting, and synthesising any message and she is only limited by the constraints of the deployed cryptographic methods.~We enrich this adversarial model by considering ``compromised MEC servers''.~This is to say that the adversary per se could be inside a \emph{legitimate MEC server} interacting with the cluster-head by using valid credentials and having \emph{privileged access} to MEC servers.~In this way, the adversary can inject a fake \textsc{Reply}, crafted with \emph{malware}, and send it back to the data requestor aiming at infecting her device.

\subsection{Malware detection}
In this adversarial environment, we envisage the use of anti-malware controls running in each device.~These can be responsible for \emph{scanning network traffic} for patterns to detect known malicious attempts.~Each device may even respond to newly detected attack methods (anomaly-based detection).~Upon detection, devices can block messages that are likely to consist of insecure content preventing, in this way, the spread of malware to other devices within their cluster.~This assumption can be seen as an advanced application of the \emph{next-generation firewalls} to mobile devices.~Although in this paper we assume that any detected malice is blocked by the device that has successfully undertaken the inspection, our work can be extended to support collaborative (e.g.~reputation-based) filtering towards blocking messages that end up having a bad reputation.~Such an approach can take advantage of \emph{learning} techniques and its investigation will be part of our future work.

\begin{figure}[t]
\centering
\includegraphics[width=380pt]{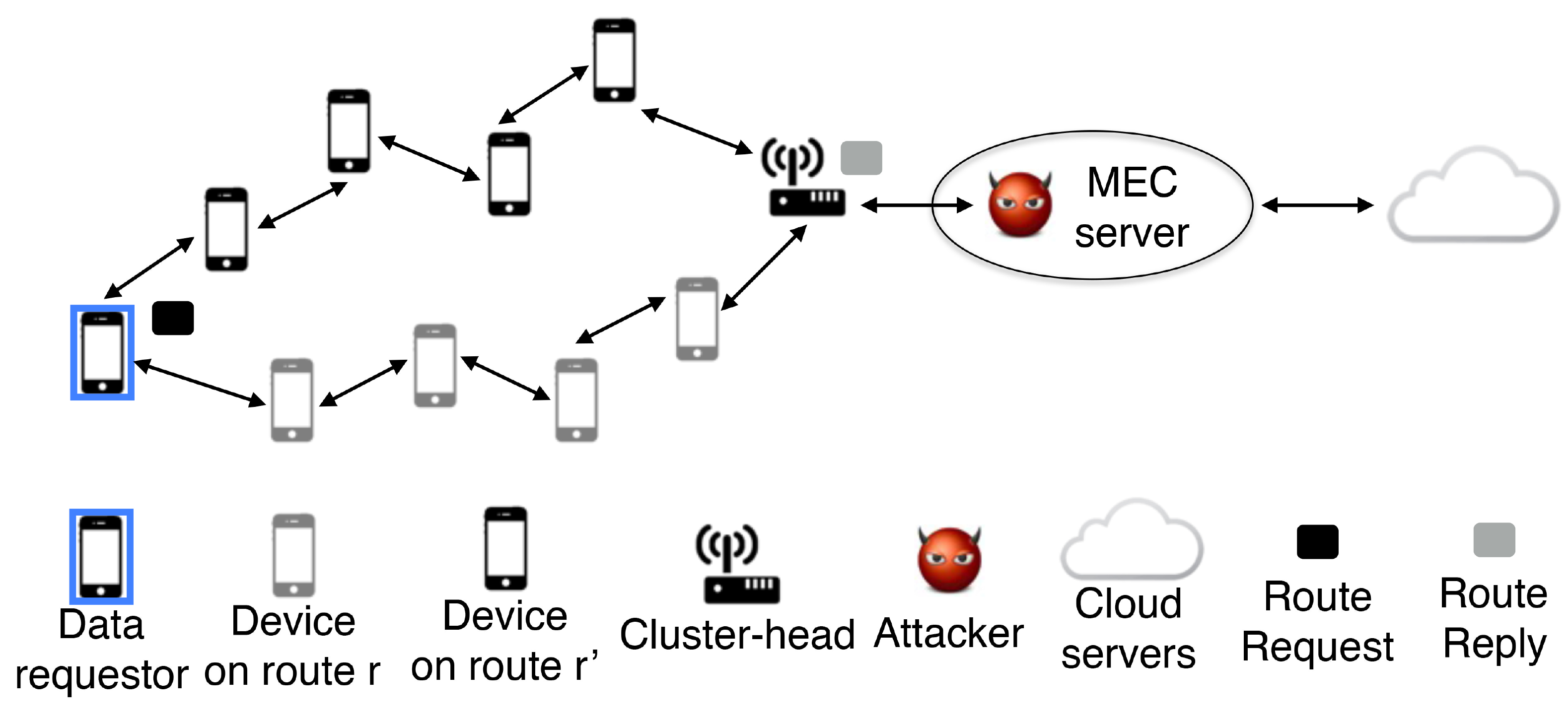}
\caption{Investigated system model, where a device requests data, that the cluster devices do not possess, from the MEC server.~The adversary has successfully launched a MITM attack controlling the communication between cluster-head and MEC server.}
\label{fig:network}
\end{figure}

\subsection{Formulation}
Let us assume a cluster of $N$ devices.~We denote by $\clhead$ its cluster-head, and by $\destination$ the requestor of some data.~Henceforth we will refer to this data as $\data$.~If the latter can not be found within the cluster itself, $\destination$ must seek $\data$ hosted by the MEC servers of its cluster.~Thus, $\clhead$ receives a \textsc{Request} from $\destination$, and it then queries the MEC server.

When $\clhead$ receives back a \textsc{Reply} from the MEC server and $\destination$ is not within its transmission range, a route $r$ must be established to deliver $\data$ from $\clhead$ to $\destination$.~Therefore, there is a need for the devices to relay $\data$ towards $\destination$, but before that, $\clhead$ \emph{must decide upon $r$}.~We assume $R$ routes available between $\clhead$ and $\destination$, we denote by $r_j \in [R]$, the $j$th route, and the set of devices that constitute $r_j$ are expressed by $\mc{S}_j$.~Note that we use the notation $[\Xi]$ to represent the set of $\Xi$ elements.

Although the route selection can be entirely taken based on quality-of-service parameters optimising network delay and jitter, the presence of an Attacker, let it be $\attacker$, introduces uncertainty with regards to the malice of the data conveyed toward $\destination$.~For instance, if $\attacker$ controls the link $\clhead \Longleftrightarrow$~MEC, then $\data$ can be anything including malware.~If this is the case, $\destination$, which trusts $\clhead$, is very likely to be infected by this malware.~In this paper, the infection risk depends on the likelihood the malware to be collaboratively detected prior to the data being used by $\destination$.~This detection relies on devices that forward packets to $\destination$, as these are also inspecting the incoming and outgoing network traffic.~

Let us consider $\noos$ different mobile operating systems, and $M_{\os}$ different malware available to the Attacker to infect devices that run a mobile operating system $\os \in [\noos]$.~Each device may run one or more anti-malware controls and for each $\os$ we assume $AM_{\os}$ anti-malware controls, which can mitigate malware that targets devices running $\os$.

Let us also assume $S$ devices and a device $s_i \in [S]$, which runs $\os$, might have available a combination of anti-malware controls given by the set $[AM^i_{\os}] \subseteq [AM_{\os}]$.~We use the characteristic function\footnote{this is a function defined on a set $X$ that indicates membership of an element in a subset $X'$ of $X$, having the value 1 for all elements of $X'$ and the value 0 for all elements of $X$ not in $X'$.} $\vec1_{[AM^i_{\os}]}: [AM_{\os}] \rightarrow \{0,1\}$ defined as follows:
\begin{align}
\vec1_{[AM_{\os}]}(a_z) := 
\begin{cases}
1, &~\text{if}~a_z \in [AM_{\os}],\\
0, &~\text{if}~a_z \notin [AM_{\os}].
\end{cases}
\end{align}
to express whether a control $a_z$ is installed in $s_i$ or not.

We express by $d(m_l,a_z) \in [0,1)$ the effectiveness of anti-malware control $a_z$ in mitigating $m_l \in [M_{\lambda}]$.~As a device can run one or more anti-malware controls, and each control $a_z$ has $1-d(m_l,a_z)$ probability of failing to detect $m_l$, the probability of $s_i$ failing to detect $m_l$ equals
\begin{align}\label{eq:psm}
p(s_i,m_l) := \prod_{a_z \in [AM_{\os}]: \vec1_{[AM_{\os}]}(a_z)=1} [1-d(m_l,a_z)]\,.
\end{align}
Note that each control detects different signs of malware and \emph{no interdependencies}, in terms of detection capabilities, are assumed in this paper.~To put it differently, we have assumed that an anti-malware control is the minimal piece of software that detects certain malicious signs.

We define as 
\begin{equation}
\vec{p}(s_i):=[p(s_i,m_l)]_{m_l \in [M_{\os}]} \in [0,1]^{M_{\os}}.
\end{equation}
the vector of \emph{failing detection probabilities}, which captures the \emph{effectiveness} of $s_i$ on detecting malware of the set $[M_{\os}]$.~One challenge here is to be able to derive these probabilities in practice.~This, for instance, can be done by undertaking thorough penetration tests (i.e.~ethical hacking) to assess the efficacy of each anti-malware control.~These tests can be performed offline for individual software components and then their combinations can be deployed on the devices.~As a result of this we can derive the probability of $m_l$ to infect $\destination$, when $\clhead$ uses the $j$th route for data delivery, as follows:
	\begin{equation}\label{eq_routedetprob}
	p(r_j,m_l) := \prod_{s_i \in \mc{S}_j} p(s_i,m_l).
	\end{equation}
	Thus, we define as $\vec{p}(r_j):=[p(r_j,m_l)]_{m_l \in [M]}$ the vector of probabilities $r_j$ to be infected by the different malware.~For more convenience, Table 1 summarizes the notation used in this paper.

\begin{table}[t]
\centering
\footnotesize
\caption{List of Symbols}
\label{tab_symbols}
\renewcommand*{\arraystretch}{1.1}
\begin{tabular}{| c | m{4.5cm} | c | m{4.5cm} |}
\hline \textbf{Symbol} & \textbf{Description} &  \textbf{Symbol} & \textbf{Description}\\ \hline \hline
\rowcolor{gray!15}
$[N]$ & Set of $N$ devices & $\clhead$ & Cluster-head \\
$\destination$ & Data requestor & $\data$ & Requested data\\

\rowcolor{gray!15}
$[R]$ & Set of routes from $\clhead$ to $\destination$ & $r_j$ & j-th route\\
$\mc{S}_j$ & Set of devices on $r_j$ & $\attacker$ & Attacker\\

\rowcolor{gray!15}
$[\noos]$ & Set of mobile operating systems & $\os$ & Operating system\\
$[M_{\os}]$ & Set of malware that can infect $\os$ & $[AM_{\os}]$ & Set of anti-malware controls for $\os$\\

\rowcolor{gray!15}
$[S]$ & Set of devices & $s_i$ & i-th device\\
$m_l$ & l-th malware & $d(m_l,a_z)$ & Effectiveness $a_z$ in mitigating $m_l$\\

\rowcolor{gray!15}
$p(s_i,m_l)$ & Probability of $s_i$ failing to detect $m_l$ & $\vec{p}(s_i)$ & Vector of ``failing-to-detect'' probabilities of $s_i$ for different malware\\
$p(r_j,m_l)$ & Probability of $\destination$ to be infected with malware $m_l$ when $\data$ is sent over $r_j$ & $\vec{p}(r_j)$ & Vector of infection probabilities for $r_j$ and all malware types\\

\rowcolor{gray!15}
$[M]$ & Set of malware & $\mixedstratD$ & Defender's mixed strategy\\

$\mixedstratA$ & Attacker's mixed strategy & $S(r_j,m_l)$ & Expected security damage on route $r_j$ when relaying $m_l$\\

\rowcolor{gray!15}
$c(s_i)$ & Malware detection cost on $s_i$ & $C(r_j)$ & Malware detection cost on $r_j$\\
$H(m_l)$ & Security loss inflicted by $m_l$ &
$L$ & path length\\

\rowcolor{gray!15}
$\mc{C}_j$ & Set of computational malware inspection costs $c(s_i)$ in $r_j$ &
$\mc{T}_j$ & Set of  malware inspection capabilities $\vec{p}(s_i)$ in $r_j$\\

\hline
\end{tabular}
\end{table}

\subsection{Game model}\label{sec:game_model}
Now that we have defined our system model by describing its components and their relationship,~in the rest of this section, we use game theory to investigate the optimal strategic routing decisions of $\clhead$,~the Defender,~and the Attacker who aims to infect one of the cluster devices with mobile malware.~The Attacker's objective is to succeed an attack against $\destination$ and the Defender must select a route to deliver the $\reply$ to $\destination$.

We define the \emph{Malware Detection Game} (MDG) between Defender and Attacker,~as an \emph{one-shot,~bimatrix} game of \emph{complete information} played for each requestor that seek some data.~The set of pure strategies of the Defender consists of all possible routes,~$r_j \in [R]$,~from $\clhead$ to $\destination$.~On the other hand,~the pure strategies of the Attacker are the different malware $m_l \in [M]$ that can be injected into the D2D network in the form of a \textsc{Reply}.~Thus,~in MDG a pure strategy profile is a pair of Defender and Attacker actions,~$(r_j,m_l) \in [R] \times [M]$ giving a pure strategy space of size $R \times M$.~For the rest of the paper,~the convention is adopted where the Defender is the row player and the Attacker is the column player.

Each player's preferences are specified by her \emph{payoff function},~and we define as $U_d: (r_j,m_l) \rightarrow \mathbb{R_{-}}$ and $U_a: (r_j,m_l) \rightarrow \mathbb{R_{+}}$ the payoff functions of the Defender and Attacker,~respectively,~when the pure strategy profile $(r_j,m_l)$ is played.~According to \cite{osborne1994course},~we define a \emph{preference relation} $\succsim$,~when $m_l $ is chosen by the Attacker,~by the condition $r_x \succsim r_y$,~if and only if $U_d(r_x,m_l) \geq U_d(r_y,m_l)$.~In general,~given the set $[R]$ of all available routes from $\clhead$ to $\destination$,~a rational Defender can choose a route (i.e.~pure strategy) $r^*$ that is \emph{feasible},~that is $r^* \in [R]$,~and \emph{optimal} in the sense that $r^* \succsim r,~\forall~r\in [R],~r\neq r^*$; alternatively she solves the problem $\max_{r\in [R]} U_d(r,~m_l)$,~for a message $m_l \in [M]$.~Likewise,~we can define the preference relation for the Attacker,~where $m_x \succsim m_y \iff U_a(r_j,m_x)\geq U_a(r_j,m_y)$,~for a route $r_j \in [R]$.

MDG can be seen as a \emph{game per session}, where the start of each session is signified by the transmission of a new \textsc{Reply} that the cluster-head will send to $\destination$; it is also realistic to assume that over a time period, there will be multiple sessions.~To derive optimal strategies for the Defender during the repetitions of MDGs, we deploy the notion of \emph{mixed strategies}.~Since players act independently,~we can enlarge their strategy spaces,~so as to allow them to base their decisions on the outcome of random events that create uncertainty to the opponent about individual strategic choices maximising their payoffs.~Hence,~both Defender and Attacker deploy randomised (i.e.~mixed) strategies.~The mixed strategy $\mixedstratD$ of the Defender is a probability distribution over the different routes (i.e.~pure strategies) from $\clhead$ to $\destination$,~where $\mixedstratD(r_j)$ is the probability of delivering a \textsc{Reply} via $r_j$ under mixed strategy $\mixedstratD$.~We refer to a mixed strategy of the Defender as a \emph{Randomised Delivery Plan} (RDP).~For the finite nonempty set $[R]$,~let $\Pi_{[R]}$ represent the set of all probability distributions over it,~i.e.
\begin{eqnarray}
\Pi_{[R]} := \{\mixedstratD \in \mathbb{R}^{+R} |~\sum_{r_j \in [R]} \mixedstratD(r_j)=1 \}.	
\label{eq:set_probs_R}
\end{eqnarray}
Therefore a member of $\Pi_{[R]}$ is a mixed strategy of the Defender.

Likewise,~the Attacker's mixed strategy is a probability distribution over the different available malware.~This is denoted by $\mixedstratA$,~where $\mixedstratA(m_l)$ is the probability of choosing $m_l$ under mixed strategy $\mixedstratA$.~We refer to a mixed strategy of the Attacker as the \emph{Malware Plan} (MP).~Similarly with (\ref{eq:set_probs_R}),~we express by $\Pi_{[M]}$ the set of all probability distributions over the set of all Attacker's pure strategies given by $[M]$.~Thus,~a member of $\Pi_{[M]}$ is as a mixed strategy of the Attacker.~From the above,~the set of mixed strategy profiles of MDG is the Cartesian product of the individual mixed strategy sets,~$\Pi_{[R]} \times \Pi_{[M]}$.	

\begin{definition}
The support of RDP $\mixedstratD$ is the set of routes $\{r_j|\mixedstratD(r_j)>0\}$,~and it is denoted by $supp(\mixedstratD)$.
\end{definition}

\begin{definition}
The support of MP $\mixedstratA$ is the set of malware $\{m_l|\mixedstratA(m_l)>0\}$,~and it is denoted by $supp(\mixedstratA)$.
\end{definition}

The above definitions state that the subset of routes (resp.~malware) that are assigned positive probability by the mixed strategy $\mixedstratD$ (resp.~$\mixedstratA$) is called the \emph{support} of $\mixedstratD$ (resp.~$\mixedstratA$).~Note that a pure strategy is a special case of a mixed strategy,~in which the support is a single action.

Now that we have defined the mixed strategies of the players,~we can define MDG as the finite strategic game $\Gamma=\1 (\mathrm{Defender},~\mathrm{Attacker}),~\Pi_{[R]} \times \Pi_{[M]},~(U_d,U_a)\2$.
For a given mixed strategy profile $(\mixedstratD,\mixedstratA) \in \Pi_{[R]} \times \Pi_{[M]}$,~ we denote by $U_d(\mixedstratD,\mixedstratA)$,~and $U_a(\mixedstratD,\mixedstratA)$ the expected payoff values of the Defender and Attacker,~where the expectation is due to the independent randomisations according to mixed strategies $\mixedstratD$,~and $\mixedstratA$.

Formally 
\begin{equation}\label{eq:util_def}
\begin{aligned}
U_d(\mixedstratD,\mixedstratA) := \sum_{r_j\in [R]} \sum_{m_l\in [M]} U_d(r_j,m_l)  \,\mixedstratD(r_j) \,\mixedstratA(m_l).~
\end{aligned}	
\end{equation}

and similarly 
\begin{equation}\label{eq:util_att}
\begin{aligned}
U_a(\mixedstratD,\mixedstratA) :=  \sum_{r_j\in [R]} \sum_{m_l\in [M]} U_a(r_j,m_l) \,\mixedstratD(r_j) \,\mixedstratA(m_l).~\
\end{aligned}
\end{equation}

By using the preference relation we can say that,~for an Attacker's mixed strategy $\mixedstratA$,~the Defender prefers to follow the RDP $\mixedstratD$ as opposed to $\mixedstratD'$ (i.e.~$\mixedstratD \succsim \mixedstratD'$),~if and only if $U_d\mixed \geq U_d(\mixedstratD',\mixedstratA)$.

\begin{definition}
The Defender's $($resp.~Attacker's) best response to the mixed strategy $\mixedstratA$ $($resp.~$\mixedstratD$) of the Attacker (resp.~Defender) is a RDP $\mixedstratD^{\BR} \in \Pi_{[R]}$ $($resp.~$\mixedstratA^{\BR} \in \Pi_{[M]})$ such that $U_d(\mixedstratD^{\BR},\mixedstratA) \geq U_d(\mixedstratD,\mixedstratA),~\forall~\mixedstratD\in \Pi_{[R]}$ $($resp.~$U_a(\mixedstratD,\mixedstratA^{\BR}) \geq U_d(\mixedstratD,\mixedstratA),~\forall~\mixedstratA\in \Pi_{[M]})$.
\end{definition}

It is noteworthy to mention that the game theoretic solutions that we will propose, in the next section, involve \emph{randomisation}.~For instance, in a mixed equilibrium,~each player's randomisation leaves the other \emph{indifferent} across her randomisation support.~These choices can be deliberately randomised or be taken by software agents that run in mobile devices (i.e.~cluster-heads or adversaries).~However these are not the only equilibria interpretations.~For instance,~the probabilities over the pure actions (i.e.~route or malware pure selections) can represent (i) time averages of an ``adaptive'' player,~(ii) a vector of fractions of a ``population'',~where each player type adopts pure strategies and,~(iii) a ``belief'' vector that each player has about the other regarding their behaviour.

\section{Game solutions}\label{solutions}

Now that we have defined MDG along with its components, in this section we concentrate in deriving optimal strategies for the Defender.~First,~we investigate the problem of determining best RDPs and MPs (i.e.~mixed strategies),~for the Defender and the Attacker respectively,~when both parties are rational decision-makers and they play simultaneously.~Note that a \emph{game solution} is a prediction of how rational players may take decisions.

As we have not explicitly defined the \emph{strategic type} of Attacker, we consider different types of solutions based on various Attacker behaviours.~This analysis will allow us to draw robust conclusions regarding the \emph{overall optimal} Defender strategy, which will minimise expected damages \emph{regardless of the Attacker type}.

\subsection{Nash mixed strategies}
The most commonly used solution concept in game theory is that of \emph{Nash Equilibrium} (NE).~This concept captures a steady state of the play of the MDG in which Defender and Attacker hold the correct expectation about the other players' behaviour and they act rationally.~In other words,~an NE dictates optimal responses to each other's actions,~keeping the others' strategies fixed,~i.e.~strategy profiles that are resistant against unilateral deviations of players.

\begin{definition}
In any Malware Detection Game (MDG),~a mixed strategy profile $(\mixedNE)$ of $\Gamma$ is a mixed NE if and only if 
\begin{enumerate}
	\item $\mixedstratD^{\NE} \succsim \mixedstratD,~\forall \mixedstratD \in \Pi_{[R]}$,~when the Attacker chooses $\mixedstratA^{\NE}$,~i.e.
	\begin{eqnarray}
		U_d(\mixedNE)\geq_{\forall \mixedstratD\in \Pi_{[R]}} U_d(\mixedstratD,\mixedstratA^{\NE});
	\end{eqnarray} 
	\item $\mixedstratA^{\NE}\succsim \mixedstratA,~\forall \mixedstratA \in \Pi_{[M]}$,~when the Defender chooses $\mixedstratD^{\NE}$,~i.e.~
	\begin{eqnarray}
		U_a(\mixedNE)\geq_{\forall \mixedstratA\in \Pi_{[M]}} U_a(\mixedstratD^{\NE},\mixedstratA).
	\end{eqnarray} 
\end{enumerate}
\end{definition}

\begin{definition} 
The Nash Delivery Plan (NDP),~denoted by $\mixedstratD^{\NE}$,~is the probability distribution over the different routes,~as determined by the NE of the MDG.
\end{definition} 

For instance,~a NDP $(0.7,0.3)$ dictates that 70\% of the $\reply$s will be sent over $r_1$,~and 30\% over $r_2$.~Note that this distribution does not determine which $\reply$ is sent over which route,~as this decision is probabilistic.

\subsection{Maximin strategies}

We say that the Defender maximinimizes if she chooses an RDP that is best for her on the assumption that whatever she does,~the Attacker will choose an MP to cause the highest possible damage to her.

\begin{definition}
A Randomised Delivery Plan $\maximinD \in \Pi_{[R]}$ is a maximin strategy of the Defender,~if and only if 
\begin{eqnarray}
\min_{\mixedstratA \in \Pi_{[M]}} U_d(\maximinD,\mixedstratA) \geq \min_{\mixedstratA \in \Pi_{[M]}} U_d(\mixedstratD,\mixedstratA),\forall \mixedstratD \in \Pi_{[R]}.
\end{eqnarray}
\end{definition}

A maximinimiser for the Defender is an RDP that maximises the payoff that the Defender can \emph{guarantee}.~In other words,~$\mixedstratD^{\dagger}$ guarantees (i.e.~``secures'') the Defender at least her maximin payoff regardless of $\mixedstratA$,~as $\mixedstratD^{\dagger}$ solves the problem $\max_{\mixedstratD} \min_{\mixedstratA} U_d(\mixedstratD,\mixedstratA)$.~That is why $\mixedstratD^{\dagger}$ is also called \emph{security strategy}.

\begin{definition}
A Malware Plan $\maximinA \in \Pi_{[M]}$ is a maximin strategy of the Attacker,~if and only if 
\begin{eqnarray}
\min_{\mixedstratD \in \Pi_{[R]}} U_a(\mixedstratD,\maximinA) \geq \min_{\mixedstratD \in \Pi_{[R]}} U_a(\mixedstratD,\mixedstratA),\forall \mixedstratA \in \Pi_{[M]}.
\end{eqnarray}
\end{definition}

\subsection{Stackelberg mixed strategies}
A \emph{two-player Stackelberg game} involves one player (leader) to commit to a strategy before the other player (follower) moves.~In a Stackelberg model the \emph{commitment of the leader is absolute},~that is the leader cannot back-track on her commitment.~On the other hand,~the follower sees the strategy that the leader committed to,~before she chooses a strategy.

In an Stackelberg MDG,~the Attacker \emph{conducts surveillance} before she attacks and therefore she is aware of the Defender's RDP.~For completeness,~we consider that this best-response is expressed also in mixed strategies.

In general,~Stackelberg and Nash games \emph{do not have the same equilibria}.~For instance,~let us consider the normal-form MDG in Table \ref{tab_example},~where the Defender has only two routes ($r,r'$) available and the Attacker can choose between two malware types ($m,m'$).~We see that if this is a Nash game,~$r$ is a strictly dominant strategy for the Defender,~as it gives her a higher payoff value than $r'$.~As we have assumed that this is a complete information game,~the Attacker knows that $r$ is preferable for the Defender and she chooses $m$,~which rewards her with 1 as opposed to $m'$,~which gives payoff value 0.~Therefore the NE of the game (in pure strategies) is $(r,m)$.

If we now consider this game as Stackelberg,~the Defender (leader) can commit to a strategy before the Attacker (follower) chooses her strategy.~If the Defender commits to $r$ then the Attacker will play $m$,~but if the Defender commits to $r'$ then the Attacker will choose $m'$.~The second pure strategy profile,~i.e.~$(r',m')$ gives higher payoff to the Defender (-2 as opposed to $(r,m)$,~which gives -3) and therefore the Defender is better-off in the Stackelberg game compared to the Nash game,~where her payoff equals -3 $<$ -2.

\begin{table}
\centering
\caption{A toy game example}
\label{tab_example}
\renewcommand*{\arraystretch}{1}
\begin{tabular}{| c | c | c |}
\hline & $m$ & $m'$\\
\hline $r$ & -3,1 & -1,0\\
\hline $r'$ & -4,0 & -2,1\\
\hline
\end{tabular}
\end{table}

\begin{definition}
A Reply Delivery Plan (RDP) is optimal if it maximises the Defender's payoff given that the Attacker will always play a best-response strategy with tie-breaking in favour of the Defender.
\end{definition}

\begin{definition}
 A Malware Plan is a best response if it maximises the Attacker's payoff,~taking the Defender's Reply Delivery Plan as given.
\end{definition}

A commonly used notion of a solution in Stackelberg games is the Strong Stackelberg Equilibrium (SSE),~defined in MDG as follows.

\begin{definition} 
At the Strong Stackelberg Equilibrium of the MDG:
\begin{enumerate}
\item for any $\mixedstratD \in \Delta_{[R]}$,~the Attacker plays a best-response $\mixedstratA^{\BR}(\mixedstratD) \in \Delta_{[M]}$ that is,
\begin{equation}\label{eq:sse_cond1}
\medmuskip=0mu
\thinmuskip=0mu
\thickmuskip=0mu
\begin{aligned}
U_a(\mixedstratD,\mixedstratA^{\BR}(\mixedstratD)) \geq
  U_a(\mixedstratD,\mixedstratA(\mixedstratD)),~\forall \,\mixedstratA(\mixedstratD) \neq \mixedstratA^{\BR}(\mixedstratD);
\end{aligned}	
\end{equation}

\item for any $\mixedstratD \in \Delta_{[R]}$,~the Attacker breaks ties in favour of the Defender,~that is,~when there are multiple best responses to $\mixedstratD$,~the Attacker plays the best response $\mixedstratA^{\SSE}(\mixedstratD) \in \Delta_{[M]}$ that maximises the Defender's payoff:
\begin{equation}\label{eq:sse_cond2}
\medmuskip=0mu
\thinmuskip=0mu
\thickmuskip=0mu
\begin{aligned}
U_d(\mixedstratD,\mixedstratA^{\SSE}(\mixedstratD)) \geq
  U_d(\mixedstratD,\mixedstratA^{\BR}(\mixedstratD)),\\ 
  \forall \mixedstratA^{\BR}~\text{best response to}~\mixedstratD;
\end{aligned}	
\end{equation}
\item the Defender plays a best-response $\mixedstratD^{\SSE}\in \Delta_{[R]}$,~which maximises her payoff given that the Attacker's strategies are given by the first two conditions (i.e.~the Attacker always plays best response with tie-breaking in favour of the Defender \cite{tambe2011},\cite{kiekintveld2009computing}):
\begin{equation}\label{eq:sse_cond3}
\medmuskip=0mu
\thinmuskip=0mu
\thickmuskip=0mu
\begin{aligned}
U_d(\mixedstratD^{\SSE},\mixedstratA^{\SSE}(\mixedstratD^{\SSE})) \geq U_d(\mixedstratD,~ \mixedstratA^{\SSE}(\mixedstratD)),~\forall \,~\mixedstratD \neq \mixedstratD^{\SSE}.
\end{aligned}	
\end{equation}
\end{enumerate}
\end{definition}

\section{Optimality analysis}\label{analysis}

For the purpose of analysis, we consider \emph{complete information} Nash MDGs, according to which both players know the game matrix, which contains the utilities of both players for each pure strategy profile.~The utility function of the Defender is determined by the probability of failing to detect a route and the overall performance cost, which is imposed on the devices of the $j$-th route when undertaking malware detection.~We denote by $c(s_i)$ the performance cost imposed on each $s_i \in \mc{S}_j$ and therefore the overall performance cost over a route $r_j$ equals $\sum_{s_i\in \mc{S}_j} c(s_i)$.

We consider two different MDGs; (i) a \emph{zero sum} MDG, where the Attacker's utility is the opposite of the Defender's utility and (ii) a \emph{non-zero sum} MDG, where the Attacker's utility is a strictly positive scaling of the Defender's utility.

The rationale behind the zero sum game is that when there are clear winners (e.g.~the Attacker) and losers (e.g.~the Defender), and the Defender is uncertain about the Attacker type, she considers the \emph{worst case scenario}, which can be formulated by a zero sum game where the Attacker can cause her \emph{maximum damage}.~While in most security situations the interests of the players are neither in strong conflict nor in complete identity, the zero sum game provides important insights into the notion of ``optimal play'', which is closely related to the \emph{minimax theorem} \cite{minimax}.

In the zero sum MDG, $\zs=\1 \{d,a\}, [R] \times [M], \{U_d,-U_d\}\2$ (for clarity $d$ has been used for the Defender and $a$ for the Attacker), the Attacker's gain is equal to the Defender's security loss, and vice versa.~We define the utility of the Defender in $\zs$ as 
\begin{eqnarray}\label{eq:utility_defender_in_zs}
\small
U_d^{\zs}(r_j,m_l) \dfn - w_{H}\,p(r_j,m_l)\,H(m_l) - w_{C}\sum_{s_i\in \mc{S}_j} c(s_i).
\end{eqnarray}
The first term of (\ref{eq:utility_defender_in_zs}) is the expected security loss of the Defender inflicted by the Attacker when attempting to infect $\destination$ with $m_l$, while the second term expresses the aggregated message inspection cost imposed on all devices of $r_j$, irrespective of the attacking strategy.~Note that $w_H, w_C \in [0,1]$ are importance weights, which can facilitate the Defender with setting her  preferences in terms of security loss, and computational detection cost, accordingly.

By setting $S(r_j,m_l)=w_{H}\,p(r_j,m_l)\,H(m_l)$, and $C(r_j)=w_{C}\sum_{s_i\in \mc{S}_j} c(s_i)$, we have that
\begin{equation} 
\label{eq:utility_defender_in_zs_short}
U_d^{\zs}(r_j,m_l)\dfn - S(r_j,m_l) - C(r_j).~
\end{equation}

For a mixed profile $(\mixedstratD,\mixedstratA)$, the utility of the Defender equals
\begin{equation}\label{eq:mixed_payoff_def}
\small
\begin{aligned}
&U_d^{\Gamma_0}(\mixedstratD,\mixedstratA) \overset{(\ref{eq:util_def})}{=}
\sum_{r_j\in [R]} \sum_{m_l\in [M]} U_d^{\Gamma_0}(r_j,m_l)  \mixedstratD(r_j) \, \mixedstratA(m_l) \\
&\overset{(\ref{eq:utility_defender_in_zs_short})}{=}
\sum_{r_j\in [R]} \sum_{m_l\in [M]} [- S(r_j,m_l) - C(r_j)]\,\mixedstratD(r_j) \,\mixedstratA(m_l) \\
&=- \sum_{r_j\in [R]} \sum_{m_l\in [M]} S(r_j,m_l)\,\mixedstratD(r_j) \,\mixedstratA(m_l)\\
&- \sum_{r_j\in [R]} C(r_j)\,\mixedstratD(r_j).
\end{aligned}
\end{equation}

As $\zs$ is a zero sum game, the Attacker's utility is given by $U_a^{\zs}(\mixedstratD, \mixedstratA) = -\,U_d^{\zs}(\mixedstratD, \mixedstratA)$.~Since the Defender's equilibrium strategies maximise her  utility, given that the Attacker maximises her  own utility, we will refer to them as \emph{optimal strategies}.

As $\zs$ is a two-person zero sum game with finite number of actions for both players, according to Nash \cite{Nash:NAS:1950}, it admits at least a NE in mixed strategies, and saddle-points correspond to Nash equilibria as discussed in \cite{Alpcan:Book:2012} (p.\,42).~The following result from \cite{Basar:Book:1995}, establishes the existence of a saddle (equilibrium) solution in the games we examine and summarizes their properties.

\begin{definition}[Saddle point of the MDG] 

The $\zs$ Malware Detection Game (MDG) admits a saddle point in mixed strategies, $(\mixedstratD^{\NE}_{\zs},\mixedstratA^{\NE}_{\zs})$, with the property that
\begin{itemize} 
\item $\mixedstratD^{\NE}_{\zs}=\arg \max_{\mixedstratD \in \Delta_{[R]}} \min_{\mixedstratA \in \Delta_{[M]}} U_d^{\Gamma_0}(\mixedstratD,\mixedstratA),\; \forall \mixedstratA$, and
\item $\mixedstratA^{\NE}_{\zs}=\arg \max_{\mixedstratA \in \Delta_{[M]}} \min_{\mixedstratD \in \Delta_{[R]}} U_a^{\Gamma_0}(\mixedstratD,\mixedstratA),\,\forall \mixedstratD$.
\end{itemize}
Then, due to the zero sum nature of the game, the minimax theorem \cite{minimax} holds, i.e.~
$\max_{\mixedstratD \in \Delta_{[R]}} \min_{\mixedstratA \in \Delta_{[M]}} U_d^{\Gamma_0}(\mixedstratD,\mixedstratA)= \min_{\mixedstratA \in \Delta_{[M]}} \max_{\mixedstratD \in \Delta_{[R]}} U_d^{\Gamma_0}(\mixedstratD,\mixedstratA)$.

The pair of saddle point strategies $(\mixedstratD^{\NE}_{\zs},\mixedstratA^{\NE}_{\zs})$ are at the same time security strategies for the players, i.e.~they ensure a minimum performance regardless of the actions of the other.~Furthermore, if the game admits multiple saddle points (and strategies), they have the ordered interchangeability property, i.e.~the player achieves the same performance level independent from the other player's choice of saddle point strategy.
\end{definition}

The minimax theorem \cite{minimax} states that for zero sum games, NE and minimax solutions coincide.~Therefore, $\mixedstratD^{\NE}_{\zs} = {\tt \arg\min}_{\mixedstratD \in \Delta_{[R]}} \max_{\mixedstratA \in \Delta_{[M]}} U_a^{\Gamma_0} (\mixedstratD,\mixedstratA)$.~This means that regardless of the strategy the Attacker chooses, the Nash Delivery Plan (NDP) is the Defender's security strategy that guarantees a minimum performance.

We can convert $\zs$ into a Linear Programming (LP) problem and make use of some of the powerful algorithms available for LP to derive the equilibrium.~For a given mixed strategy $\mixedstratD$ of the Defender, we assume that the Attacker can cause maximum damage to $\destination$ by injecting a message $\widehat{m}$ into the cluster network.

Formally, the Defender seeks to solve the following LP: 
\begin{equation}
\small
\begin{aligned}
&\max_{\mixedstratD \in \Delta_{[R]}} \min_{\mixedstratA \in \Delta_{[M]}} U_d^{\Gamma_0}(\mixedstratD, \widehat{m}\,) \\ 
\text{subject} & \text{~to}
\begin{cases} 
\label{eq:lp}
U_{d}^{\Gamma_0}(\mixedstratD,m_1) - \min_{\mixedstratA \in \Delta_{[M]}}U_d^{\Gamma_0}(\mixedstratD, \widehat{m})e \geq 0\\ \hspace{2cm} \vdots \\ U_{d}^{\Gamma_0}(\mixedstratD,m_M) - \min_{\mixedstratA \in \Delta_{[M]}}U_d^{\Gamma_0}(\mixedstratD, \widehat{m})e \geq 0\ \\ \mixedstratD e = 1 \\ \mixedstratD \geq 0.~
\end{cases} 
\end{aligned}
\end{equation}

In this problem, $e$ is a vector of ones of size $M$.

\begin{lemma}
\label{lemma:mixed_ne}
A mixed strategy profile $(\mixedstratD^{\NE},\mixedstratA^{\NE}) \in \Pi_{[R]} \times \Pi_{[M]}$ in $\Gamma_0$, is a mixed strategy NE if and only if
\begin{enumerate}
	\item every route $r_j \in supp(\mixedstratD^{\NE})$ selection is a best response to $\mixedstratA^{\NE}$ and,
	\item every malware $m_l \in supp(\mixedstratA^{\NE})$ selection is a best response to $\mixedstratD^{\NE}$.
\end{enumerate}
\end{lemma}
\begin{proof}
First, notice that $U_d$, as defined in (\ref{eq:utility_defender_in_zs}), is a linear function in $\mixedstratD(r_j)$ that is, for any two RDPs $\mixedstratD_1$ and $\mixedstratD_2$ and any number $\theta \in [0,1]$ we have $U_d(\theta\,\mixedstratD_1 + (1-\theta)\,\mixedstratA)=\theta\,U_d(\mixedstratD_1)+(1-\theta)\,U_d(\mixedstratD_2)$.~Then, for the sake of contradiction, assume there exists a route $r'_j \in supp(\mixedstratD^{\NE})$ selection that is not a best response to $\mixedstratA^{\NE}$.~Due to the linearity of $U_d$ in $\mixedstratD^{\NE}(r_j)$, the Defender can increase her payoff by transferring probability from $\mixedstratD(r'_j)$ to a route selection that is a best response to $\mixedstratA^{\NE}$, creating a new mixed strategy $\mixedstratD^* \succsim \mixedstratD^{\NE}$.~However, this contradicts the assumption that $\mixedstratD^{\NE}$ is the strategy of the Defender at the NE, as the Defender prefers to deviate from $\mixedstratD^{\NE}$ to gain a higher payoff, by playing $\mixedstratD^*$.~The second part of the lemma can be proven in the same way.~
\end{proof}

Let us now assume a non-zero sum MDG, denoted by $\nzs$, with the same strategy spaces with $\zs$, in which the Defender's utility is the same as in $\zs$, i.e.~$U_d^{\nzs}(\mixedstratD,\mixedstratA)=U_d^{\zs}(\mixedstratD,\mixedstratA)=-S(r_j,m_l)-C(r_j)$.~On the other hand, the Attacker's utility is (strictly positive) scaling of the security loss $S(r_j,m_l)$ of the Defender upon a successful attack.~This is to say that the performance cost of the Defender is only important to her  as the Attacker is only after compromising $\destination$.~Therefore, given a pure strategy profile $(r_j,m_l)$, the utility of the Attacker, in $\nzs$, is defined as: 
\begin{eqnarray}
\label{eq:util_att_zero_sum}
U_a^{\nzs}(r_j,m_l) \dfn \mathrm{\Xi}\,S(r_j,m_l),~\text{for}~\mathrm{\Xi}>0.	
\end{eqnarray}

For a mixed profile $(\mixedstratD,\mixedstratA)$ the utility of the Attacker is given by 
\begin{equation}\label{eq:mixed_payoff_att}
\begin{aligned}
U_a^{\nzs}(\mixedstratD,\mixedstratA) &\overset{(\ref{eq:util_att})}{=}
\sum_{r_j\in [R]} \sum_{m_l\in [M]} U_a^{\nzs}(r_j,m_l) \, \mixedstratD(r_j) \, \mixedstratA(m_l) \\
&\overset{(\ref{eq:util_att_zero_sum})}{=}
\sum_{r_j\in [R]} \sum_{m_l\in [M]} \Xi \, S(r_j,m_l)\,\mixedstratD(r_j) \, \mixedstratA(m_l).
\end{aligned}
\end{equation}

Hence, due to $U_d^{\nzs}(\mixedstratD,\mixedstratA)=U_d^{\zs}(\mixedstratD,\mixedstratA)$, from (\ref{eq:mixed_payoff_def}) and (\ref{eq:mixed_payoff_att}) we have that
\begin{equation}\label{eq:util_att_def}
\small
\begin{aligned}
U_d^{\nzs}(\mixedstratD,\mixedstratA) &=- \frac{1}{\Xi} U_a^{\nzs}(\mixedstratD,\mixedstratA) - \sum_{r_j\in [R]} C(r_j)\,\mixedstratD(r_j) \\
&= - \frac{1}{\Xi} U_a^{\nzs}(\mixedstratD,\mixedstratA) - k(\mixedstratD),
\end{aligned}
\end{equation}
where $\frac{1}{\Xi}>0$, and $k(\mixedstratD)$ is an expression that does not depend on $\mixedstratA$.~That is, the best response of the Defender to any given malware plan, also yields the utility for the Defender at the worst case scenario.~

\begin{lemma}
\label{lemma:nes}
NE strategies of the Defender in $\nzs$ are equivalent of the NE strategies of the Defender in $\zs$.~Formally, $\setne_{\nzs}=\setne_{\zs}$.
\end{lemma}
\begin{proof}
By definition, a strategy profile $(\mixedNE)$ is NE of $\nzs$ if and only if:
\begin{subequations}
\small
\begin{align}
S(\mixedNE)+k(\mixedstratD^{\NE})\leq 
S(\mixedstratD,\mixedstratA^{\NE})+k(\mixedstratD), \forall 
\mixedstratD\in\Delta_{[R]}, \\
\mathrm{\Xi}\cdot S(\mixedNE)\geq 
\mathrm{\Xi}\cdot S(\mixedstratD^{\NE},\mixedstratA), \forall 
\mixedstratA\in\Delta_{[M]}.
\end{align}
\end{subequations}
Here is the observation:
\begin{equation}
\begin{aligned}
\mathrm{\Xi}\cdot S(\mixedNE)\geq
\mathrm{\Xi}\cdot S(\mixedstratD^{\NE},\mixedstratA), \forall 
\mixedstratA\in\Delta_{[M]}\iff \\
\mathrm{\Xi}\cdot [S(\mixedNE)+k(\mixedstratD^{\NE})]\geq \\
\mathrm{\Xi}\cdot [S(\mixedstratD^{\NE},\mixedstratA)+k(\mixedstratD^{\NE})], 
\forall \mixedstratA\in\Delta_{[M]}.
\end{aligned}
\end{equation}

Since $\mathrm{\Xi}>0$, the latter condition is satisfied if and only if:
\begin{equation}
\small
S(\mixedNE)+k(\mixedstratD^{\NE})\geq 
 S(\mixedstratD^{\NE},\mixedstratA)+k(\mixedstratD^{\NE}), \forall 
\mixedstratA\in\Delta_{[M]}.
\end{equation}
In short, $(\mixedNE)$ is a NE of $\nzs$, if and only if 
it satisfies:
\begin{subequations}
\medmuskip=0mu
\thinmuskip=0mu
\thickmuskip=0mu
\begin{align}
S(\mixedNE)+k(\mixedstratD^{\NE})\leq 
S(\mixedstratD,\mixedstratA^{\NE})+k(\mixedstratD), \forall 
\mixedstratD\in\Delta_{[R]}, \\
S(\mixedNE)+k(\mixedstratD^{\NE})\geq 
S(\mixedstratD^{\NE},\mixedstratA)+k(\mixedstratD^{\NE}), \forall 
\mixedstratA\in\Delta_{[M]}.
\end{align}
\end{subequations}
But these are exactly the conditions describing a NE of 
$\zs$.~Therefore $\setne_{\nzs}=\setne_{\zs}$.
\end{proof}

\begin{lemma}
\label{lemma:nonzerosum_ne_maximin}
In $\nzs$, the set of NE and Maximin strategies of the Defender are equivalent, i.e.~$\setne_{\nzs}=\setmaximin_{\nzs}$.~
\end{lemma}
\begin{proof}
	($\Rightarrow$)~Since $\zs$ is a two person zero-sum game, we know that the set of NE and Maximin strategies of the Defender are the same, i.e.~$\setne_{\zs}=\setmaximin_{\zs}$.~Let $(\mixedNE) \in \setne_{\nzs}$ then based on Lemma \ref{lemma:nes} we have that $(\mixedNE) \in \setne_{\zs}$.~Since $\zs$ is zero-sum, $\mixedstratD^{\NE} \in \setmaximin_{\zs}$.~But the strategy spaces and the utility of the Defender are the same in both $\nzs$ and $\zs$.~Hence the conditions for a mixed strategy to be a Defender's Maximin is the same in both games.~Therefore, $\mixedstratD^{\NE} \in \setmaximin_{\nzs}$, i.e.~$\setne_{\nzs}\subseteq\setmaximin_{\nzs}$.~\\
($\Leftarrow$)
The argument goes in the other direction as well: consider $\mixedstratD^{\NE} \in \setmaximin_{\nzs}$.~Since the utility of the Defender and the strategy spaces are the same across the two games, for the same strategy $\mixedstratD^{\NE}$, we have that $\mixedstratD^{\NE} \in \setmaximin_{\zs}$.~Since $\zs$ is two-player zero-sum, there exists $\mixedstratA^{\NE}$ such that $(\mixedNE) \in \setne_{\zs}$.~From Lemma \ref{lemma:nes}, this means $(\mixedNE)_{\Gamma} \in \setne$.~Hence, \textit{Maximin strategies of the Defender are also part of her NE strategies in $\nzs$}, i.e.~$\setmaximin_{\nzs}\subseteq \setne_{\nzs}$.~Putting the two together $\setne_{\nzs}=\setmaximin_{\nzs}$.
\end{proof}

This lemma establishes that the Defender can randomise according to her  NE and, in expectation, be guaranteed at least the expected utility prescribed by the NE, irrespective of the mixed strategy of the Attacker.~To put it differently, the Defender can play her  pessimistic maximin strategy, but she does not lose anything in expectation by not playing a NE strategy.~It is worth stressing that this property only holds for the NE strategy of the Defender and not of the Attacker.

\begin{lemma}
\label{lemma:Maxmin_SSE}
In $\nzs$, the set of Maximin and SSE strategies of the Defender are the same, i.e.~$\setmaximin_\nzs=\setsse_\nzs$.
\end{lemma}
\begin{proof}
($\Rightarrow$) Let $\mixedstratD^{\NE} \in \setsse_\nzs$  be a SSE strategy of the Defender.~Then by definition, $\mixedstratD^{\NE} $ is (i) an optimal strategy of the Defender given that (ii) the Attacker is best-responding to it but by (iii) breaking ties in favour of the Defender.~That is:
\begin{enumerate}
	\item[(i)]~$\mixedstratD^{\NE} \in\arg\max_{\mixedstratD\in\Delta_{[R]}} U_d(\mixedstratD,\mixedstratA^{\BR}(\mixedstratD))$ where;
	\item[(ii)]~for any $\mixedstratD\in\Delta_{[R]}$, $\mixedstratA^{\BR}(\mixedstratD)\in\arg\max_{\mixedstratA\in\Delta_{[M]}} U_a(\mixedstratD,\mixedstratA)$ and;
	\item[(iii)]~for any $\mixedstratD\in\Delta_{[R]}$:
	\begin{equation}
	\small
	\mixedstratA^{\BR}(\mixedstratD)\in\arg\max_{\mixedstratA\in\arg\max_{\mixedstratA\in\Delta_{[M]}} U_a(\mixedstratD,\mixedstratA)} U_d(\mixedstratD,\mixedstratA).
	\end{equation}
\end{enumerate}
 Let us examine condition (ii): for any $\mixedstratD\in\Delta_{[R]}$: 
 \begin{equation}
 \small
\begin{aligned}
\mixedstratA^{\BR}(\mixedstratD)\in\arg\max_{\mixedstratA\in\Delta_{[M]}} \mathrm{\Xi}\cdot S(\mixedstratD,\mixedstratA) \iff \\ \mixedstratA^{\BR}(\mixedstratD)\in\arg\max_{\mixedstratA\in\Delta_{[M]}} \mathrm{\Xi}\cdot [S(\mixedstratD,\mixedstratA)+k(\mixedstratD)]\\
 \mixedstratA^{\BR}(\mixedstratD)\in\arg\max_{\mixedstratA\in\Delta_{[M]}} S(\mixedstratD,\mixedstratA)+k(\mixedstratD).
\end{aligned}
\end{equation}
In short, condition (ii) is equivalent to: 
\begin{equation*}
\label{cond_2'}	
\text{(iv)~For any}~\mixedstratD\in\Delta_{[R]}, \mixedstratA^{\BR}(\mixedstratD)\in\arg\min_{\mixedstratA\in\Delta_{[M]}} U_d(\mixedstratD,\mixedstratA).
\end{equation*}
This makes condition (iii) irrelevant.~But conditions (i) and (iv) exactly describe a Maximin strategy of the Defender.~Therefore we have proved that $\setsse_{\nzs}\subseteq\setmaximin_{\nzs}$.~($\Leftarrow$) The argument can be established identically in reverse direction, starting from a Maximin strategy of the Defender.~So given conditions (i) and (iv) we must prove that conditions (ii) and (iii) are true.~Let $\mixedstratD^{\NE} \in \setmaximin_\nzs$ be a Maximin strategy of the Defender.~Then by definition, $\mixedstratD^{\NE}$ is (i) an optimal strategy of the Defender given that (iv) the Attacker is minimising Defender's utility.~~We see that condition (ii) is true if and only if condition (iv) is true.~Since the Maximin strategy $\mixedstratD^{\NE}$ makes condition (iv) true, it will also make condition (ii).~To prove that $\mixedstratD^{\NE}$ is an SSE, we also need to prove condition (iii).~Let us assume that the condition is not true.~This means that there is a best-response of the Attacker that does not break ties in favour of the Defender.~Formally, 
\begin{equation}
\medmuskip=0mu
\thinmuskip=0mu
\thickmuskip=0mu
\begin{aligned}	
\mixedstratA^{\BR}(\mixedstratD)\notin\arg\max_{\mixedstratA\in\arg\max_{\mixedstratA} U_a(\mixedstratD,\mixedstratA)} U_d(\mixedstratD,\mixedstratA) \iff \\
\mixedstratA^{\BR}(\mixedstratD)\notin\arg\max_{\mixedstratA\in\arg\max_{\mixedstratA} U_a(\mixedstratD,\mixedstratA)}\big\{-S(\mixedstratD,\mixedstratA)-k(\mixedstratD)\big\} \iff \\
\mixedstratA^{\BR}(\mixedstratD)\notin\arg\min_{\mixedstratA\in\arg\max_{\mixedstratA} U_a(\mixedstratD,\mixedstratA)}\big\{S(\mixedstratD,\mixedstratA)+k(\mixedstratD)\big\}\iff \\
\mixedstratA^{\BR}(\mixedstratD)\notin\arg\min_{\mixedstratA\in\arg\max_{\mixedstratA} U_a(\mixedstratD,\mixedstratA)}S(\mixedstratD,\mixedstratA)\iff \\
\mixedstratA^{\BR}(\mixedstratD)\notin\arg\min_{\mixedstratA\in\arg\max_{\mixedstratA} U_a(\mixedstratD,\mixedstratA)} U_a(\mixedstratD,\mixedstratA),
\end{aligned}
\end{equation}	
which is leads to a contradiction.~Therefore condition (3) holds, and putting together all three conditions (1), (2), and (3), we have that $\mixedstratD^{\NE}$, which is a Maximin strategy of the Defender it is also an SSE strategy, i.e.~$\setmaximin_\nzs\subseteq\setsse_\nzs$.~Putting the two proofs together we have that $\setmaximin_\nzs=\setsse_\nzs$.
\end{proof}

\begin{thm}
\label{theorem:ne_sse_maximin}
In $\nzs$, the set of NE, Maximin and SSE strategies of the Defender are the same,~i.e.~$\setne_{\nzs}=\setmaximin_{\nzs}=\setsse_{\nzs}$.~Besides, all NE are interchangeable, in $\nzs$, and all yield the same utility for the defender.
\end{thm}
\begin{proof}
Trivially, from Lemmas \ref{lemma:nonzerosum_ne_maximin} and \ref{lemma:Maxmin_SSE} we have that $\setne_{\nzs}=\setmaximin_{\nzs}=\setsse_{\nzs}$.~Since $\zs$ is a two person zero-sum game, we know that all NE are interchangeable \cite{Basar:Book:1995}.~From Lemma \ref{lemma:nes} the NE of $\zs$ are the NE of $\nzs$ and vice-versa.~We also see that the utility of the Defender is the same across $\nzs$ and $\zs$.~Therefore the utility of the Defender in all NE of our original game is the same, which also implies that all NE of our original game are interchangeable.
\end{proof}

The above lemma establishes that the Defender, regardless of whether the Attacker conducts surveillance, she plays optimally when she randomises according to her  NE strategy.~

\begin{thm}
\label{theorem:zerosum_ne_optimal}
Regardless of the type of malware detection game played, i.e.
\begin{enumerate}
\item a zero sum or a non-zero sum malware detection game,
\item a Nash or a Stackelberg malware detection game,
\end{enumerate}
the Defender plays optimally by choosing any strategy $\mixedstratD \in \setne_{\zs}$.~
\end{thm}
\begin{proof}
By combining \ref{lemma:nes} and \ref{theorem:ne_sse_maximin}, we have that 
$\setne_{\zs}=\setne_{\nzs}=\setmaximin_{\nzs}=\setsse_{\nzs}$, which proves the theorem.
\end{proof}

The above theorem demonstrates that it is computationally efficient for the Defender to derive her  optimal strategy by solving the LP represented by (\ref{eq:lp}).~It is worth noting that a similar result but for different problem has been published in \cite{korzhyk2011}.

\section{$i$Routing}\label{irouting}
In this section, we present the $i$Routing protocol, which stands for \emph{intelligent Routing} and whose routing decisions are made according to the \emph{Nash Delivery Plan} (NDP).~$i$Routing has been designed based on the mathematical findings of the MDG analysis, presented in previous sections, and its main goal is to maximise the utility of the Defender in the presence of a ``rational'' Attacker.

Within the realm of Mobile Edge Computing (MEC), devices of the cluster request services from the cluster-head (denoted by $\clhead$) imposing the need for establishing an end-to-end path between the requestor (i.e.~destination device denoted by $\destination$) and $\clhead$.~Each time data must be delivered to $\destination$, $\clhead$ has to compute the NDP by solving an MDG for this destination.~To do this, following the route discovery, $\clhead$ uses its latest information about the malware detection capabilities of all possible routes to $\destination$, along with their inspection costs (i.e.~malware detection costs to perform, for example, intrusion classification).~Data is then relayed and collaboratively inspected by the devices on its way to $\destination$.~Overall, the objective of $\clhead$ (i.e.~the Defender) is to select the route that can correctly detect and filter out malicious data before they infect $\destination$ by making sure that it is not crafted with malware.~We assume that each device must use its data inspection capabilities at the maximum possible degree..

$i$Routing has characteristics of \emph{reactive route selection protocols}, meaning that it takes action and starts computing routing paths that have not been previously computed when a request for data delivery to $\destination$ is issued.~$i$Routing requires to obtain information about the malware inspection capabilities and the associated computational cost of devices, in routes from $\clhead$ to $\destination$.

\begin{algorithm}
\begin{algorithmic}[1]
\Procedure{$i$Routing\_Request}{$s,\destination,\mc{S}_j$}
	\State $s$ seeks routes to $\destination$ by broadcasting $\tt RREQ_{\destination}$;
	\If{a device $s_i$ receives $\tt RREQ_{\destination}$}
		\State $\mc{S}_j \cup \{s_i\}$;
		\If{$s_i \neq \destination$}
			\State $s_i$ executes \textsc{$i$Routing\_Request}(${s_i,\destination},\mc{S}_j$);
		\Else
			\State $L \gets |\mc{S}_j|$, $n\gets0,\mc{T}_j\gets \emptyset,\mc{C}_j\gets \emptyset$;
			\State \textsc{$i$Routing\_Response}($n,L,\mc{T}_j,\mc{C}_j,\mc{S}_j,\destination$);
			\State break;
		\EndIf
	\EndIf
\EndProcedure
\caption{Seeking routes to destination $\destination$.}
\label{alg_irouting_request}
\end{algorithmic}
\end{algorithm}

$i$Routing consists of \emph{three main phases}, which we describe in more detail in the remainder of this section.~In the first phase of the protocol (described in Algorithm \ref{alg_irouting_request}),~$\clhead$ \emph{broadcasts} a Route REQuest ($\tt RREQ_{\destination}$) to discover routes towards $\destination$.~Each device that receives the $\tt RREQ_{\destination}$), acts similarly by broadcasting it towards $\destination$.~After $\clhead$ sends a $\tt RREQ_{\destination}$, it has to await for some timeout $T_{req}$, which is set equal to the Net Traversal Time (NetTT), as in AODV \cite{Perkins:2003:AODV:RFC3561}.

The second phase of the protocol starts when the receiving device is $\destination$.~Then, this device does not forward the request any further.~Instead, it prepares a Route REPly ($\tt RREP_{\destination}$), and sends it back towards $\clhead$ by using the reverse route, which is built during the delivery of $\tt RREQ_{\destination}$, as described by Algorithm \ref{alg_irouting_response}.~Each $\tt RREP_{\destination}$ carries information about: (i) the set $\mc{S}_j$ of devices that comprise a route; (ii) the set $\mc{T}_j$ of vectors of ``failing-to-detect'' probabilities, for different malware, of devices in $r_j$; and (iii) the set $\mc{C}_j$ of computational malware inspection costs $c(s_i)$ of devices in $r_j$.~These values are updated while the $\tt RREP_{\destination}$ is traveling back to $\clhead$.~When each device (e.g.~$s_i$) that is involved in the route response phase, receives the $\tt RREP_{\destination}$, it updates $\mc{T}_j$ and $\mc{C}_j$.~Within the time period $T_{req}$, $\clhead$ aggregates $\tt RREP_{\destination}$ messages and updates its routing table with information that can be used to derive the \emph{optimal routing strategy}, as dictated by Theorem \ref{theorem:zerosum_ne_optimal}.

\begin{algorithm}[t]
\begin{algorithmic}[1]
\Procedure{$i$Routing\_Response}{$n,L,\mc{T}_j,\mc{C}_j,\mc{S}_j,s$}
\State $s$ sends $\tt RREP_{\destination}$ to the ($L-n$)-th device of $\mc{S}_j$, let it be $s_i$; 
	\If{$s_i \ne \clhead$}
	   \State $\mc{T}_j \cup \vec{p}(s_i)$, $\mc{C}_j \cup c(s_i)$, $n\gets n+1$;	
		\State \textsc{$i$Routing\_Response}($n,L,\mc{T}_j,\mc{C}_j,\mc{S}_j,s_i$);	
	\Else
		\State Execute \textsc{$i$Routing}(${\destination},\data,\mc{S}_j,\mc{T}_j,\mc{C}_j$);	
		\State break;
	\EndIf
\EndProcedure
\caption{Responding to a cluster-head with a route to $\destination$.}
\label{alg_irouting_response}
\end{algorithmic}
\end{algorithm}

\begin{algorithm}[t]
\begin{algorithmic}[1]
\Procedure{$i$Routing}{${\destination},\data,\mc{S}_j,\mc{T}_j,\mc{C}_j$}
	\State $\clhead$ derives the \emph{Nash Delivery Plan}, $\mixedstratD^{\NE}$ using $\mc{S}_j,\mc{T}_j,\mc{C}_j$;
	\State $\clhead$ chooses $r^*$ probabilistically as dictated by $\mixedstratD^{\NE}$; 
	\State $\clhead$ delivers $\data$ to $\destination$ over $r^*$; 
	\State Each device $s_i \in r^*$ performs data inspection;
	\If{$\data$ found to carry malware}
		\State $s_i$ drops $\data$;
		\State $s_i$ notifies $\clhead$ by sending a notification message along the reverse path;
		\State $\clhead$ blacklists the device that sent, through the cloud, $\data$ consisting of malware;
	\Else
		\State $s_i$ forwards $\data$ to $\destination$;
	\EndIf
\EndProcedure
\caption{Delivering data to $\destination$.}
\label{alg_irouting}
\end{algorithmic}
\end{algorithm}

In the third phase of the protocol, described in Algorithm \ref{alg_irouting}, $\clhead$ uses its routing table to solve the MDG by computing the \emph{Nash Delivery Plan}, denoted by $\mixedstratD^{\NE}$, which has a lifetime $T$.~Then, $\clhead$ probabilistically selects a route according to $\mixedstratD^{\NE}$ to deliver the requested data to $\destination$.~The chosen route is denoted by $r^*$.~Note that for the same $\destination$ and before $T$ expires, $\clhead$ uses the same $\mixedstratD^{\NE}$ to derive $r^*$, upon a new \textsc{Request}.

Also, the third phase focuses on detecting malware injected along with the requested data (denoted by $\data$) to prevent the infection of $\destination$.~While $\data$ is delivered to $\destination$ over $r^*$, the relay devices, on $r^*$, perform data inspection auditing $\data$ for malware.~Upon successful detection, the device that detects the malware, first drops $\data$, and then notifies $\clhead$ that $\data$ was crafted with malware.~The notification message is sent along the reverse path.~When receiving this, $\clhead$ blacklists the device that has originally sent $\data$ (this device is assumed that has hijacked the communication link between MEC server and the cluster-head).~This can be seen as the first step towards mitigating the investigated attack model and anything beyond that is out of the scope of this paper.

While each data $\data$ is collaboratively inspected by the devices on its way to $\destination$, the derivation of the \emph{optimal routing strategy}, i.e.~the Nash Delivery Plan (NDP), is computed only by $\tt{C}$ through solving a Malware Detection Game (MDG) for this specific destination $\destination$.~Therefore, even if the other devices are aware of the existence of some infected data, it is only $\tt{C}$ that isolates the Attacker (i.e.~data source) towards mitigating future malware infection risks.

The communications complexity of the $i$Routing protocol measured in terms of number of messages exchanged in performing route discovery is $\mathcal{O}(2N)$, where $N$ is the number of devices in the D2D network.~As a reactive routing protocol, \textit{i}Routing has higher storage complexity than conventional routing protocols, but it supports multiple-path routing and QoS routing making malware detection optimal, as shown in section \ref{analysis}.~Finally, \textit{i}Routing has a time complexity equal to $\mathcal{O}(2D)$, where $D$ is the diameter of the D2D network.

\section{Simulations}
\label{simulations}

\subsection{Network setup}
We have conducted a series of simulations to evaluate the performance of the optimal strategies in D2D networks.~Devices have been randomly deployed inside a rectangular area of 1000m x 1000m.~For each device, the transmission power is fixed, and the maximum transmission range is 200m, while two devices can directly communicate with each other only if they are in each other’s transmission range.~We have performed the simulations using the OMNeT++ network simulator and INET framework.~We have simulated the IEEE 802.11 MAC layer protocol and devices send UDP traffic.~In the simulations, the requestor of some data is chosen randomly, and the total number of devices of a \emph{cluster} is set to be 20.~The total simulation time varies (10,~20,~40,~60,~120 seconds) to confirm the consistency of results.~Table \ref {tab_parameters} summarizes the simulation parameters.

\begin{table}[t]
\centering
\footnotesize
\caption{Simulation parameter values}
\label{tab_parameters}
\renewcommand*{\arraystretch}{1.2}
\begin{tabular}{| l | p{4.5cm} |}
\hline \textbf{Parameter} & \textbf{Value} \\ \hline \hline
	Number of nodes  &  20  \\
	\rowcolor{gray!15}
	Mobility model  &Linear Mobility  \\
	Mobility Speed  & 10 m/s   \\ 
	\rowcolor{gray!15}
	Mobility Update Interval &0.1 s \\
	Packet size& 512 bytes\\
	\rowcolor{gray!15}
	Packet generation rate & 2 packets/s\\
	Simulation time & 600 s\\
	\hline
\end{tabular}
\end{table}

\subsection{Security controls and malware}
Simulations consider one adversary who is injecting a sequence of consecutive malicious replies with the aim to infect $\destination$.~We assume that the Attacker chooses to inject one of $[M]=\{$Keylogger, SMS spam, Rootkit iSAM, Spyware, iKee-B, Premium-Rate calls$\}$ malware types (i.e.~pure strategies of the Attacker).~We have also assumed the anti-malware controls, SMS Profiler, iDMA, iTL, and Touchstroke, along with their detection rates, as published in \cite{damopoulos2014}.~Each mobile device is equipped with at least one and up to three anti-malware controls.

\subsection{Attackers}
We have simulated 3 different Attacker types; namely \emph{Uniform}, \emph{Weighted}, and \emph{Nash} Attacker:
\begin{itemize}
\item \emph{Uniform}: the Attacker chooses each malware type from the set with equal probability.~For example for the set we have used here, there is a probability $\frac{1}{6}=0.1667$ the Attacker to choose any of the malware types of $[M]$;
\item \emph{Weighted}: the Attacker chooses a malware type with probability derived by the following algorithm:
\begin{enumerate}
	\item find the average utility value of the Attacker for each column of the game matrix;
	\item add the average utility values of the Attacker for all columns to get the combined sum;
	\item for each malware type, derive the probability of a malware type to be chosen by dividing its average utility value, found in step 1, by the sum derived in step 2.
\end{enumerate}
\item \emph{Nash}: the Attacker plays according to her Nash strategy $\mixedstratA^{NE}$.
\end{itemize}
Per $\reply$, the simulator chooses an attack sample from the attack probability distribution which is determined by the Attacker profile.

We have introduced different probability distributions for each Attacker type, only for testing purposes.~Nevertheless, $i$Routing is optimal regardless of the probability distribution of a malware type to be chosen by the Attacker; a petition that is formally consolidated by the mathematical results presented in sections~\ref{solutions} and \ref{analysis} as well as the simulation results uncovered in this section.

\subsection{Experiments}
We have considered 5 \emph{Cases} each referring to different simulation times: 10, 20, 40, 60, and 120 mins.~For each Case we have simulated 1,000 replies, which are UDP messages of length 512 bytes with delay limit 100 seconds, for a fixed network topology.~ Yet we refer to the run of the code for the pair $\1\mbox{Case},\#\mbox{replies}\2$ by the term \emph{Experiment}.~We have repeated each Experiment for 10 independent network topologies to get a clear idea of the results' trend.~We do that for all 5 Cases and each type of Attacker profile.~Thus we simulate, in total: $5~\mbox{Cases} \times 1,000~\mbox{replies}  \times~\mbox{10 network topologies} = 50,000~\mbox{replies}$.~

\subsection{Comparisons}
We compare $i$Routing against AODV, DSR, and custom-made routing protocol called \emph{Proportional Routing} (PR), for different Attacker types.

PR is computed as follows.~First, by using the game matrix, the Defender computes the average utility value for each row, let it be 
\begin{equation}
\hat{U}_d(r_j) = \frac{\sum_{m_l=1}^{M} U_d(r_j,m_l)	}{M},~\forall~r_j\in[R].
\end{equation}
Then, the probability of route $r_j$ to be chosen equals:
\begin{equation}
1 - \frac{\hat{U}_d(r_j)}{\sum_{r=1}^{R}\hat{U}_d(r)}.
\end{equation}

According to the results illustrated in Figures~\ref{fig_rate_nash_attacker} - \ref{fig_rate_weighted_attacker}, $i$Routing consistently outperforms the rest of the protocols, in terms of both Defender's \emph{expected utility} and \emph{average detection rate}, for all different simulation times and Attacker types.~The results show that $i$Routing achieves its highest average malware detection rate ($\sim$65\%) against a Uniform Attacker (non-strategic Attacker), and its worst rate against a Weighted Attacker.~In the case of a Nash Attacker, $i$Routing has almost 22\% higher detection rate than PR, 6\% than DSR, while it is twice more efficient (i.e.~$\sim$11\%) than AODV.~For a Weighted Attacker, PR behaves differently as it achieves approximately 6\% lower average detection rate than $i$Routing, in contrast to DSR and AODV, which perform worse, as opposed to the Nash Attacker case, since the difference of their average detection rate compared to $i$Routing becomes double (i.e.$\sim$12\% for DSR and 24\% for AODV).~Finally, for a Uniform Attacker, the difference, in terms of detection rate, compared to $i$Routing, is almost the same for both DSR and PR, which is approximately equivalent to 8\%.~AODV still has the worst average detection rate among all protocols by having 24\% worse rate than $i$Routing.

\begin{figure}[H]
\centering
\begin{tikzpicture}[scale=0.9]
\begin{axis}[
legend pos=north east,
enlargelimits={abs=0.5},
ybar=0pt,
bar width=0.2,
xlabel={Time (mins)},
ylabel={Detection rate (\%)},
xtick={0.5,1.5,...,5.5},
xticklabels={10,20,40,60,120},
x tick label as interval,
xmajorgrids,
ymajorgrids,
enlarge y limits=0.5,
]
\addplot+[draw=black,fill=black!70,error bars/.cd,
y dir=both,y explicit]
coordinates {
    (1,60.5) 
    (2,58.6) 
    (3,45.3) 
    (4,47)
    (5,46.6)};
\addplot+[draw=black,fill=black!50,error bars/.cd,
y dir=both,y explicit]
coordinates {
    (1,61.4)
    (2,60)
    (3,58.7) 
    (4,57.3) 
    (5,50)};
\addplot+[draw=black,fill=black!25,error bars/.cd,
y dir=both,y explicit]
coordinates {
    (1,33) 
    (2,41.2) 
    (3,43) 
    (4,43) 
    (5,45.3)};
\addplot+[draw=black,fill=black!10,error bars/.cd,
y dir=both,y explicit]
coordinates {
    (1,66.7) 
    (2,68.8) 
    (3,61.4) 
    (4,58.3) 
    (5,60)};
\legend{\footnotesize AODV,\footnotesize DSR,\footnotesize PR,\footnotesize $i$Routing}
\end{axis}
\end{tikzpicture}
\caption{Malware detection rate in presence of a Nash attacker.}\label{fig_rate_nash_attacker}
\end{figure}
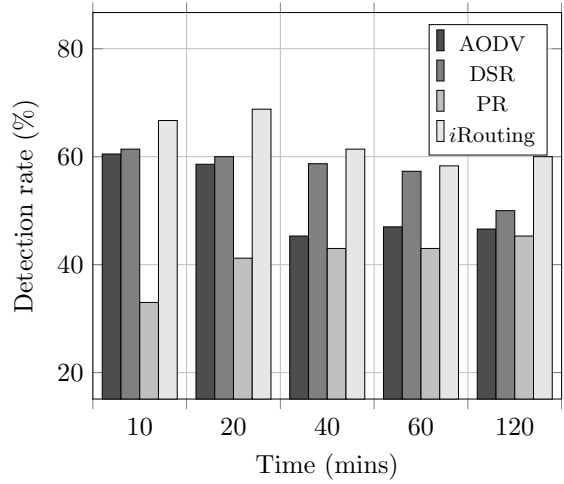

\begin{figure}[H]
\centering
\begin{tikzpicture}[scale=0.9]
\begin{axis}[
legend pos=north east,
enlargelimits={abs=0.5},
ybar=0pt,
bar width=0.2,
xlabel={Time (mins)},
ylabel={Detection rate (\%)},
xtick={0.5,1.5,...,5.5},
xticklabels={10,20,40,60,120},
x tick label as interval,
xmajorgrids,
ymajorgrids,
enlarge y limits=0.5,
]

\addplot+[draw=black,fill=black!70,error bars/.cd,
y dir=both,y explicit]
coordinates {
    (1,44.3)
    (2,38.8) 
    (3,39) 
    (4,40) 
    (5,35.5)};
\addplot+[draw=black,fill=black!50,error bars/.cd,
y dir=both,y explicit]
coordinates {
    (1,55.6) 
    (2,63) 
    (3,60) 
    (4,59.3) 
    (5,57)};
\addplot+[draw=black,fill=black!25,error bars/.cd,
y dir=both,y explicit]
coordinates {
    (1,53)
    (2,59) 
    (3,57) 
    (4,57) 
    (5,55)};
\addplot+[draw=black,fill=black!10,error bars/.cd,
y dir=both,y explicit]
coordinates {
    (1,75.7) 
    (2,70) 
    (3,62.5) 
    (4,60) 
    (5,61)};
\legend{\footnotesize AODV,\footnotesize DSR,\footnotesize PR,\footnotesize $i$Routing}
\end{axis}
\end{tikzpicture}
\caption{Malware detection rate in presence of a Uniform attacker.}\label{fig_rate_uniform_attacker}
\end{figure}
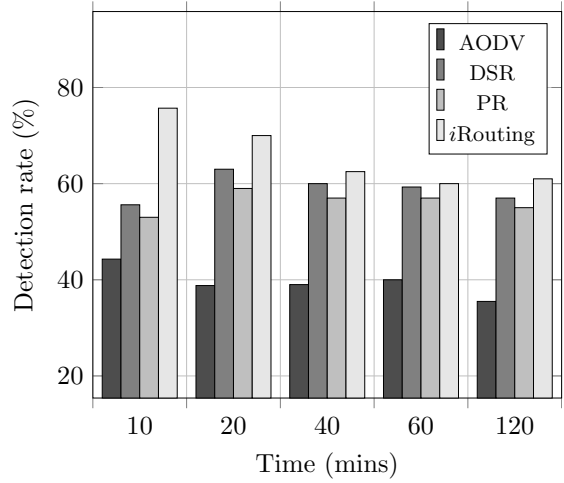

\begin{figure}[H]
\centering
\begin{tikzpicture}[scale=0.9]
\begin{axis}[
legend pos=north east,
enlargelimits={abs=0.5},
ybar=0pt,
bar width=0.2,
xlabel={Time (mins)},
ylabel={Detection rate (\%)},
xtick={0.5,1.5,...,5.5},
xticklabels={10,20,40,60,120},
x tick label as interval,
xmajorgrids,
ymajorgrids,
enlarge y limits=0.5,
]

\addplot+[draw=black,fill=black!70,error bars/.cd,
y dir=both,y explicit]
coordinates {
    (1,44.3) 
    (2,36.4) 
    (3,33.7) 
    (4,38.8) 
    (5,34)};
\addplot+[draw=black,fill=black!50,error bars/.cd,
y dir=both,y explicit]
coordinates {
    (1,55)
    (2,57)
    (3,40)
    (4,47)
    (5,50)};
\addplot+[draw=black,fill=black!25,error bars/.cd,
y dir=both,y explicit]
coordinates {
    (1,56.7)
    (2,58.7)
    (3,54.1)
    (4,50) 
    (5,57)};
\addplot+[draw=black,fill=black!10,error bars/.cd,
y dir=both,y explicit]
coordinates {
    (1,65) 
    (2,65.5)
    (3,60)
    (4,57)
    (5,60)};
\legend{\footnotesize AODV,\footnotesize DSR,\footnotesize PR,\footnotesize $i$Routing}
\end{axis}
\end{tikzpicture}
\caption{Malware detection rate in presence of a Weighted attacker.}
\label{fig_rate_weighted_attacker}
\end{figure}
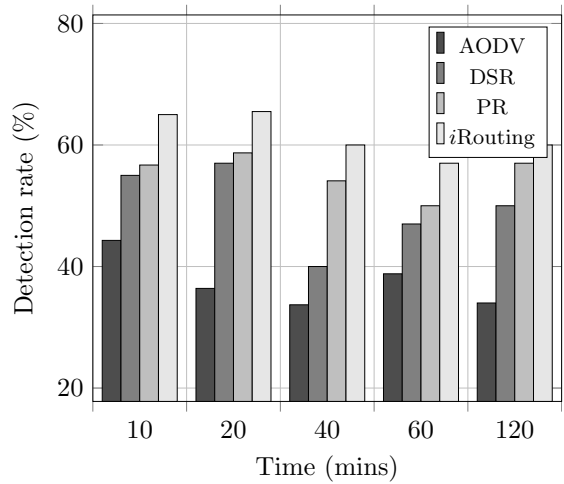

According to Figures \ref{fig_Ud_nash_attacker} - \ref{fig_Ud_weighted_attacker}, $i$Routing achieves the best performance in terms of average expected utility among all protocols.~More specifically, $i$Routing improves the average expected utility, in the case of a Nash Attacker, by, in average, 49\%, 17\%, and 7\% compared to PR, AODV, and DSR, respectively.~We notice that the Defender's utility in $i$Routing is similar to the one achieved when DSR is used.~The reason for this is that DSR improves computational cost as opposed to $i$Routing more than AODV and PR while exhibiting the best detection rate among AODV and PR.~Average improvement values are slightly more pronounced for a non-strategic Uniform Attacker; 16\%, 68\%, and 37\%, as opposed to the same protocols.~The situation is similar for a Weighted Attacker, in which case the corresponding improvement values are 18\%, 53\%, and 20\%.~We also notice that the behaviour of all protocols but $i$Routing is stochastic despite of $i$Routing having steadily the best performance.~


\begin{figure}[H]
\centering
	\begin{tikzpicture}[scale=0.6]
	\begin{axis}[title={\large Nash attacker}, thick,width=\linewidth, xlabel=time (mins), ylabel={\large $U_d$}, legend pos=south west, xmajorgrids, ymajorgrids]
	
		\addplot[dotted, very thick, blue] table[x=T, y=AODV, col sep=comma]{results/PayoffNashAtt.csv};\addlegendentry{AODV}
		
		\addplot[dashed, thick, black] table[x=T, y=DSR, col sep=comma]{results/PayoffNashAtt.csv};\addlegendentry{DSR}
		
		\addplot[dash pattern=on .8pt off 3pt on 4pt off 3pt, thick, cyan] table[x=T, y=PR, col sep=comma]{results/PayoffNashAtt.csv};\addlegendentry{PR}
		
		\addplot[red] table[x=T, y=IRR, col sep=comma]{results/PayoffNashAtt.csv};\addlegendentry{iRouting}
		
	\end{axis}
	\end{tikzpicture}
	\caption{Utility of the Defender in presence of a Nash attacker.}\label{fig_Ud_nash_attacker}
\end{figure}
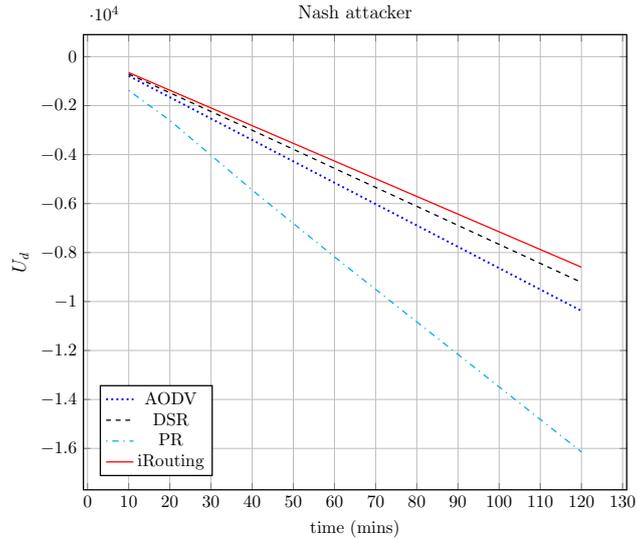

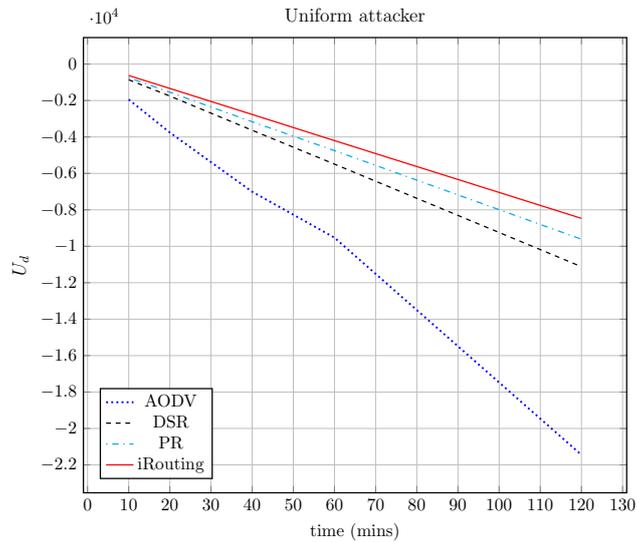
\begin{figure}[H]
\centering
	\begin{tikzpicture}[scale=0.6]
	\begin{axis}[title={\large Uniform attacker}, thick,width=\linewidth, xlabel=time (mins), ylabel={\large $U_d$}, legend pos=south west, xmajorgrids, ymajorgrids]
	
		\addplot[dotted, very thick, blue] table[x=T, y=AODV, col sep=comma]{results/PayoffUniformAtt.csv};\addlegendentry{AODV}
	
		\addplot[dashed, thick, black]
		table[x=T, y=DSR, col sep=comma]{results/PayoffUniformAtt.csv};\addlegendentry{DSR}
		
		\addplot[dash pattern=on .8pt off 3pt on 4pt off 3pt, thick, cyan]  table[x=T, y=PR, col sep=comma]{results/PayoffUniformAtt.csv};\addlegendentry{PR}
		
		\addplot[red] table[x=T, y=IRR, col sep=comma]{results/PayoffUniformAtt.csv};\addlegendentry{iRouting}
		
	\end{axis}
	\end{tikzpicture}
	\caption{Utility of the Defender in presence of a Uniform attacker.}\label{fig_Ud_uniform_attacker}
\end{figure}

\begin{figure}[H]
\centering
	\begin{tikzpicture}[scale=0.6]
	\begin{axis}[title={\large Weighted attacker}, thick,width=\linewidth, xlabel=time (mins), ylabel={\large $U_d$}, legend pos=south west, xmajorgrids, ymajorgrids]
	
		\addplot[dotted, very thick, blue] table[x=T, y=AODV, col sep=comma]{results/PayoffWeightedAtt.csv};\addlegendentry{AODV}
		
		\addplot[dashed, thick, black] table[x=T, y=DSR, col sep=comma]{results/PayoffWeightedAtt.csv};\addlegendentry{DSR}
		
		\addplot[dash pattern=on .8pt off 3pt on 4pt off 3pt, thick, cyan] table[x=T, y=PR, col sep=comma]{results/PayoffWeightedAtt.csv};\addlegendentry{PR}
		
		\addplot[red] table[x=T, y=IRR, col sep=comma]{results/PayoffWeightedAtt.csv};\addlegendentry{iRouting}
	\end{axis}
	\end{tikzpicture}
	\caption{Utility of the Defender in presence of a Weighted attacker.}\label{fig_Ud_weighted_attacker}
\end{figure}
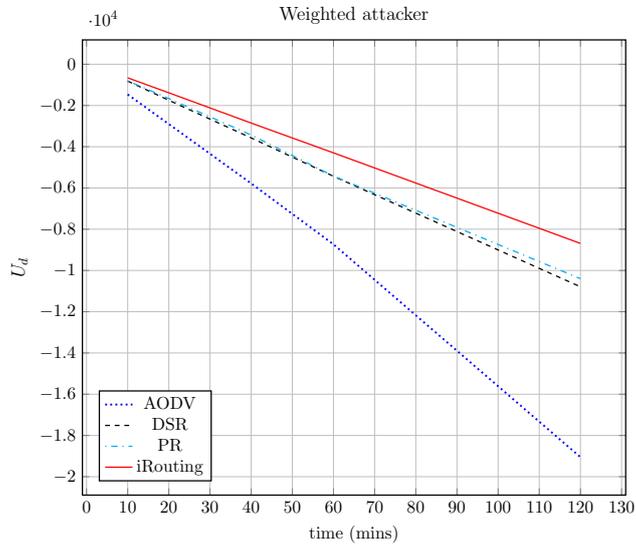

\section{Conclusion}\label{conclusion}
In this paper, we have formally investigated how to select an end-to-end path to deliver data from a source to a destination in device-to-device networks under a game theoretic framework.~We assume the presence of an external adversary who aims to infect ``good'' network devices with malware.~First, a simple yet illuminating two-player security game, between the network (the Defender) and an adversary, is studied.~To devise optimal routing strategies, optimality analysis has been undertaken for different types of games to prove, \emph{in theory}, that there is a Nash equilibrium strategy that always makes the Defender better-off.~The analysis has shown that the expected security damage that can be inflicted by the \emph{Attacker} is bounded and limited when the proposed strategy is used by the Defender.~Network simulation results have also illustrated, \emph{in practice}, that the proposed strategy can effectively mitigate malware infection.~In future work, we intend to investigate machine learning algorithms (e.g.~boosting) to convert weak learners (e.g.~devices with limited number of anti-malware controls) to strong ones.

\section{References}
\bibliographystyle{elsarticle-num} 
\bibliography{references}

\end{document}